\documentclass[journal,onecolumn]{IEEEtran}
\usepackage{cite}
\usepackage{amsmath,amssymb,amsfonts}
\usepackage{algorithmic}
 \usepackage{booktabs}
\usepackage{graphicx}
\usepackage{textcomp}
\usepackage{amsfonts,amssymb,amsmath}
\usepackage{graphicx,graphics,epsfig,color}
\usepackage{multicol}
\usepackage{float}
\usepackage[caption=false]{subfig}
\usepackage{comment}
\usepackage{caption}
\usepackage{diagbox}

\usepackage{graphicx} 
\usepackage{here}
\usepackage{amsmath}
\usepackage{amsthm}
\usepackage{amssymb}
\usepackage{amsfonts}
\usepackage{graphicx}

\usepackage{xpatch}
\xpatchcmd{\paragraph}{\normalfont}{{\normalfont\bfseries}}{}{}
\usepackage{color}
\usepackage{dsfont}

\usepackage{nicematrix}

\usepackage[ruled,vlined,linesnumbered]{algorithm2e}

\newtheorem{theorem}{Theorem}
\newtheorem{remark}{Remark}

\newtheorem{corollary}{Corollary}
\newtheorem{definition}{Definition}

\newcommand{\R}{{\mathbb{R}}}

\newcommand{\N}{{\mathbb{N}}}

\newcommand{\uu}{{\mathbf{u}}}
\newcommand{\uuu}{{\mathbf{\tilde{u}}}}
\newcommand{\dd}{{\mathbf{d}}}
\newcommand{\nex}{\mathord{\bigcirc}}

\usepackage{hyperref}

\usepackage{enumitem}
\usepackage{balance}
\definecolor{olivegreen}{rgb}{0.14,0.29,0}

\newif\ifitsdraft
\def\itsdraft{\global\itsdrafttrue}

\itsdraft  
\ifitsdraft

\definecolor{gray}{rgb}{0.33,0.4,0.47}

\definecolor{steelblue}{rgb}{0,.42,.7}

\definecolor{britishgreen}{rgb}{0,0.26,0.15}
\definecolor{navyblue}{rgb}{0,0,.8}
\definecolor{olivegreen}{rgb}{0.14,0.29,0}
\definecolor{myred}{rgb}{0.86,0.1,0.16}
    
\newcounter{al}    \newcounter{ss}

    \allowdisplaybreaks
    \pdfoutput = 1
    
     \usepackage{soul} 
\begin{document}
\title{Resilient and Effort-Optimal Controller Synthesis under Temporal Logic Specifications}
% \author{Adnane Saoud, Ryan  S. Johnson, and Ricardo G. Sanfelice, 
% \thanks{A. Saoud is with the College of Computing, University Mohammed VI Polytechnic, Benguerir, Morocco {e-mail:adnane.saoud@um6p.ma}). Ryan S. Johnson and Ricardo  G.  Sanfelice are with the Dept. of Electrical and Computer Engineering, University of California,  Santa Cruz,  CA,  USA (e-mail:ricardo@ucsc.edu,rsjohnso@ucsc.edu). \\ Research by R. G. Sanfelice partially supported by NSF Grants no. CNS-2039054 and CNS-2111688, by AFOSR Grants nos. FA9550-23-1-0145, FA9550-23-1-0313, and FA9550-23-1-0678, by AFRL Grant nos. FA8651-22-1-0017 and FA8651-23-1-0004, by ARO Grant no. W911NF-20-1-0253, and by DoD Grant no. W911NF-23-1-0158.}}

\author{Youssef Ait Si, Ratnangshu Das, Negar Monir,  Sadegh Soudjani, Pushpak Jagtap, and Adnane Saoud
\thanks{Youssef Ait Si and Adnane Saoud are with the College of Computing, University Mohammed VI Polytechnic, Benguerir, Morocco.  (e-mail: \{youssef.aitsi, adnane.saoud\}@um6p.ma) }
\thanks{Ratnangshu Das and Pushpak Jagtap are The Robert Bosch Centre for Cyber-Physical Systems, IISc, Bangalore, India (e-mail: \{ratnangshud, pushpak\}@iisc.ac.in).}
\thanks{Negar Monir is with 
the Newcastle University, Newcastle upon Tyne, United Kingdom (e-mail: S.Seyedmonir2@newcastle.ac.uk).}
\thanks{Sadegh Soudjani is with the Max Planck Institute for Software Systems, D-67663 Kaiserslautern, Germany, and with the University of Birmingham, United Kingdom (e-mail: sadegh@mpi-sws.org).}
}

\maketitle

\begin{abstract}
In this paper, we consider the notions of effort and resilience of a dynamical control system defined by the maximum disturbance the system can withstand while satisfying given finite temporal logic specifications. Given a dynamical system and a specification, the objective is to synthesize the controller such that the system satisfies the specification while maximizing its resilience, taking into account input constraints. In addition, we introduce a new metric, called the \textit{effort metric}, which characterizes the minimal input bound necessary to satisfy a given specification for a perturbed system. The problem for both metrics is formulated as a robust optimization program where the objective is to compute the maximum resilience for the system with input constraints or the minimal effort while simultaneously synthesizing the corresponding controller parameters. Moreover, we study the trade-off between resilience and effort, where we seek to maximize resilience and minimize the control effort. For linear systems and linear controllers, exact solutions are provided for the class of time-varying polytopic specifications for the closed-loop and open-loop systems. For the case of nonlinear systems, nonlinear controllers, and more general specifications, we leverage tools from the scenario optimization approach, offering a probabilistic guarantee of the solution as well as computational feasibility. Different case studies are presented to illustrate the theoretical results. 
\end{abstract}

% \begin{IEEEkeywords}
% Dynamical system, resilience, constrained control, optimal control, temporal logic.
% \end{IEEEkeywords}

\section{INTRODUCTION}

Real-world control systems must operate under disturbances, model uncertainties, and varying environmental conditions. Autonomous vehicles must maintain safe operation despite changing road conditions and sensor noise \cite{vargas2021overview}, robotic manipulators are required to perform tasks accurately in unpredictable environments \cite{kolhe2013robust}, and power grids must preserve stability under fluctuating demands, faults, and component failures \cite{zhang2021grid}. Traditional robust control approaches aim at guaranteeing stability and performance under bounded disturbances, typically assuming that admissible disturbance and input bounds are known a priori. While these methods provide strong safety and performance guarantees, they primarily focus on feasibility and robustness verification, rather than explicitly quantifying the limits beyond which a system can no longer satisfy its intended specifications.

This limitation has motivated the introduction of resilience metrics, which aim to quantify the maximum disturbance a system can tolerate while still ensuring correct behavior. Resilience is particularly relevant when dealing with temporal specifications \cite{fainekos2009robustness}, which impose complex constraints on system trajectories over time, including safety, reachability, and other temporal requirements. Temporal logic has become a powerful formalism in control synthesis \cite{kloetzer2008fully}, enabling the specification of behaviors such as “reach a target within a given time window” or “avoid unsafe regions at all times”. Although robust and optimal control methods can enforce such specifications under bounded uncertainties \cite{sadraddini2015robust,das2025spatiotemporal}, they do not directly aim at quantifying the maximum disturbance magnitude for which the specification remains satisfied.

Several notions of resilience have been investigated in the control systems literature \cite{rieger2009resilient, AitS2025}. In \cite{rieger2009resilient}, resilience is associated with the system’s ability to maintain situational awareness and functionality under adverse conditions, while \cite{zhu2011robust} proposes a holistic framework for robust and resilient control in power systems, with a focus on voltage regulation under faults and attacks. More recently, \cite{chen2023stl} introduced a resilience framework combining recoverability and durability, and \cite{monir2025AGC} demonstrated how resilience concepts can be used to derive feasible assume–guarantee contracts. 
However, most of these works focus on recovery after rare or extreme events and do not explicitly quantify the maximum admissible disturbance compatible with specification satisfaction.

The works closest in spirit to our approach are \cite{Saoud2023, saoud2024}, which deal with resilience for autonomous systems. However, moving from autonomous to controlled systems is a significant challenge. The central problem shifts from computing the maximal resilience to designing the controller that simultaneously ensures maximal resilience, adheres to input constraints, and makes the system satisfy the specification even if initially not satisfied. This integration of optimal control synthesis with input constraints represents a significant and previously unaddressed complexity, as the space of possible controllers, such as linear open-loop (state-independent), linear closed-loop (state-dependent), and nonlinear controllers, must be explored in a computationally tractable manner. Our contribution directly addresses this gap, providing a closed-form expression for computing the provably optimal constrained controller that achieves the highest possible level of system resilience.

While the resilience captures a system's intrinsic robustness, it does not account for the \emph{control input magnitude} required to achieve it. In practice, actuators are subject to strict amplitude, energy, or rate constraints. It is therefore crucial to determine whether a desired level of resilience can be attained within the available actuation limits. This observation motivates a complementary analysis: rather than fixing the input bounds and asking for the maximal tolerable disturbance (resilience), we ask: \emph{what is the minimal input bound necessary to guarantee specification satisfaction for a given level of disturbance?} We term this quantitative measure the \emph{effort metric}. Notably, when the disturbance is set to the system's maximal resilience, this metric defines the minimal control bound needed to achieve the best possible robustness.

Early works characterized the sets of admissible controls that render a system invariant or reachable \cite{gilbert1995linear,blanchini1999set}, while viability theory and constrained control directly link the existence of feasible control policies to the size of the admissible input set \cite{aubin1991viability,kolmanovsky1998theory}. Robust and constrained Model Predictive Control (MPC) frameworks explicitly enforce hard input bounds while guaranteeing constraint satisfaction against disturbances \cite{mayne2000constrained,bemporad2007robust,rawlings2017model}. A thread in these approaches is that the control bounds are treated as fixed design parameters.

Closely related is the literature on the minimal control bound, which assesses whether a given task can be accomplished with limited actuation \cite{borrelli2003constrained}, \cite{scokaert1999min}. Techniques from robust optimization provide tools to analyze feasibility under uncertainty and characterize constraint tightness \cite{ben2009robust,bertsimas2011theory}. In reachability analysis, studies have examined how invariant and reachable sets scale with input bounds, revealing trade-offs between disturbance rejection and control effort \cite{rakovic2005invariant}. These ideas extend to specification-driven control and temporal-logic synthesis, where feasibility is inherently tied to available control bounds \cite{haesaert2021robust,nikou2016controller}. 

Despite these advances, existing methods predominantly treat input bounds as given. There is a lack of a framework that explicitly quantifies the minimal control effort, in terms of the required input bound, to satisfy complex temporal specifications against worst-case disturbances. In contrast, this work elevates the input bound to a \emph{decision variable}. We introduce the effort metric as a measure of the actuation resources needed to enforce temporal logic specifications robustly, providing a direct bridge between specification complexity, disturbance levels, and necessary control bounds.

This paper builds upon these concepts by formulating both resilience and effort metrics as robust optimization problems. The objective is to design an optimal controller that guarantees the system satisfies a complex specification, ensuring all trajectories maintain a desired behavior, while simultaneously achieving one of two complementary goals, either determining the largest set of admissible disturbances for a given bounded input (resilience), or finding the minimal input constraints required to withstand a prescribed level of perturbations (effort). In addition, we define the trade-off between resilience and effort metric by assigning weight to each metric to find the optimal trade-off for the chosen weights, which we denote as the resilient-effort metric.

The contributions developed in this paper are as follows: defining the resilience metric with input constraints and the effort metric for a perturbed dynamical system as a robust optimization problem. In the particular case of a time-varying linear system with a linear closed-loop controller, we translate this robust optimization problem into a tractable optimization problem with polynomial constraints, yielding exact solutions. In the case of an open-loop controller for the resilience metric, the optimization problem reduces to a linear program, which allows for efficient and reliable computation using well-established linear programming techniques \cite{boyd2004convex}. In the general case of nonlinear systems and nonlinear controllers and a general finite-horizon specification, we leverage scenario optimization techniques to compute an approximation of the resilience and the effort metric with probabilistic guarantees. Moreover, we define a general framework for finite specifications, where finite-horizon safety and exact time reachability are considered as particular cases. Table~\ref{tab:paper_structure} summarizes the structure of the paper, which is organized into three main parts, each following the same structure for every metric. First, we show the theoretical results presented as theorems for the resilience–effort metric, considering both linear and nonlinear systems. For the linear case, we distinguish between closed-loop and open-loop controllers. The results for the individual resilience and effort metrics are then derived as corollaries.

\begin{table}[t]
\centering
\caption{Organization of theoretical results across system classes and performance metrics.}
\label{tab:paper_structure}
\small
\setlength{\tabcolsep}{4pt}
\begin{tabular}{lccc}
\toprule
 & \multicolumn{2}{c}{\textbf{LTV system}} & \textbf{Nonlinear system} \\
\cmidrule(lr){2-3} \cmidrule(l){4-4}
\textbf{Metric} & \textbf{Closed-loop} & \textbf{Open-loop} & \textbf{Nonlinear} \\
 & \textbf{Controller} & \textbf{Controller} & \textbf{Controller} \\
\midrule
Resilience-effort & Thm.~\ref{Theorem:7} & Thm.~\ref{Theorem:8} & Thm.~\ref{theorem:9} \\
Resilience & Corl.~\ref{theorem:1} & Corl.~\ref{theorem:2} &  \\
Effort & Corl.~\ref{Theorem:4} & Corl.~\ref{Theorem:5} &  \\
\bottomrule
\end{tabular}
\end{table}
The current paper extends the preliminary results presented in \cite{ait2025maximal} in the following directions.

\begin{enumerate}
    \item We define a new resilience metric with input constraints that generalizes the metric without input constraints from \cite{ait2025maximal}. Moreover, we provide an approach to compute this new metric.
    \item We introduce a novel effort metric that quantifies the minimal control input bound necessary to satisfy a specification for a disturbed system. We then define the resilience-effort metric for the trade-off between resilience and effort that quantifies the optimal disturbance and input bound for the chosen weights for resilience and effort metrics.
    \item For the computation of both the resilience and effort metrics, we also study the system under open-loop controllers. With appropriate reformulations, the corresponding robust optimization problems can be transformed into linear programs enabling exact and computationally efficient solutions.
    \item We generalize the results to linear time-varying (LTV) systems, which subsumes the results in~\cite{ait2025maximal} applicable only to time-invariant linear systems.
    \item  We provide additional details on numerical examples and two new case studies demonstrating the applicability of the proposed concepts. 
\end{enumerate}

The remainder of the paper is organized as follows. Section~\ref{sec:2} introduces the controlled system model, temporal specifications, together with the definitions of the resilience, effort, and resilience-effort metrics. Section~\ref{sec:3} addresses the computation of the resilience-effort metric for both linear time-varying systems and nonlinear systems. For linear systems, we study linear closed-loop and open-loop controllers, where exact solutions are derived. Section~\ref{sec:4} and~\ref{sec:5} develop analogous results for resilience and effort metrics, respectively, including exact formulations in the linear case and extensions to nonlinear dynamics. Section~\ref{sec:6} presents numerical results illustrating the proposed methods.
\section{Preliminaries and Problem Formulation}

\label{sec:2}
\textbf{Notations:}
The symbols $ \mathbb{N} $, $ \mathbb{N}_{\geq 0} $, $ \R$, and $\R_{\geq 0}$ denote the set of positive integers, nonnegative integers, real, and non-negative real numbers, respectively. We use $\mathbb{R}^{n \times m}$ to denote the space of real matrices with $n$ rows and $m$ columns. For a matrix $A \in \mathbb{R}^{n \times m}$, $A^T$ represents the transpose of $A$ and $A \geq 0$ denotes a matrix non-negative elements. For a vector $x \in \mathbb{R}^n$, we use $\|x\|$ and $\|x\|_{\infty}$ to denote the Euclidean and infinity norm, respectively. We denote the $n \times n$ identity matrix by $\mathbb{I}_n$, the $n \times n$ zero matrix by $0_n$, and the $n \times m$ zero matrix by $0_{n \times m}$. We denote by $\mathbf{1}_n  \in \mathbb{R}^n$ the vector whose elements are all equal to one. Given $x \in \mathbb{R}^n$ and $\varepsilon \geq 0$, we define $\Omega_{\varepsilon}(x)=\left\{z \in \mathbb{R}^n \mid\|z-x\|_{\infty} \leq \varepsilon\right\}$ and $\mathcal{B}_{\varepsilon}(x)=\left\{z \in \mathbb{R}^n \mid\|z-x\| \leq \varepsilon\right\}$. The combination of \(k \in  \N\) items chosen from \(n \in \N\) distinct items is given by the formula $\binom{n}{k} = \frac{n!}{k!(n-k)!}$, where \(n!\) is the factorial of \(n\).
\subsection{Discrete-time Dynamical Systems}
\label{system definition}
A time-varying discrete-time control system can be defined as a tuple $\Sigma = (X,U,D,f)$, where $X\subset \mathbb{R}^n$ is the state space, $U\subset \mathbb{R}^m$ is the input space,
$D\subset\mathbb{R}^n$ is the disturbance space, which is assumed to be a compact set and contains the origin, and $f:  \N \times X \times U \rightarrow  X$ is a continuous map representing the system dynamics. 
The system $\Sigma$ is evolving according to the following dynamics:
\begin{equation}
\label{eqn:non linear controlled system}
x(k+1) = f(k, x(k) , u(k)) + d(k),\quad  k \in \mathbb{N},
\end{equation}
where $x(k) \in X$, $u(k)  \in U $, and $d(k)\in D$ represent the system state, system input, and the additive disturbance, respectively, at time $k$. 
In this work, the system $\Sigma$ is controlled by a state feedback controller $\pi_\alpha$ defined by $\pi_\alpha : X \to U$ such that \( \pi_\alpha(x(k)) =  u(k)\), where $ \alpha \in \R^{d}$ represent the collection of parameters of the controllers. A simple example is a linear controller, which can be modeled by the function $\pi_\alpha(x) = \alpha_1 x + \alpha_2$, where $(\alpha_1, \alpha_2) \in \mathbb{R}^{m\times n} \times \mathbb{R}^{m}$ are two real matrices of dimension $m \times n$ and $m \times 1$, respectively. In this context, the set of parameters $\alpha$ is given by $\alpha = ( \alpha_1, \alpha_2) \in \R^{m\times n + m}$. Another notable example is polynomial controllers \cite{polynomial2013}. 
Consider the control function \(\pi_\alpha(x): X \rightarrow U\) defined as a polynomial of degree \(l\) given by \(\pi_\alpha(x) = \alpha \phi(x)\), where \(\phi(x) = \left[ 1, x^{[1]}, x^{[2]}, \ldots, x^{[l]} \right]\) is a vector of monomials of degree up to \(l\), with each \(x^{[i]} \in \mathbb{R}^{d(n,i)}\) containing all distinct monomials of degree \(i\) with no repeated elements. The dimension of \(x^{[i]}\) is \(d(n,i) = \binom{n+i-1}{n-1}\) and the total number of distinct monomials up to degree \(l\) is \(D(n,l) = \sum_{i=0}^{l} d(n,i) = \binom{n+l}{n}\), which is the dimension of the vector \(\phi(x)\). In this setup, we have $\alpha = [\alpha_0, \alpha_1, \ldots, \alpha_l]$ representing a vector of dimension $ m \times D(n,l)$, where $\alpha_i \in \mathbb{R}^{m \times d(n, i)}$.
Moreover, we consider the case of an open-loop controller where the control actions are determined by a time-varying sequence \( \uu =  (u_0, \dots, u_{N-1}) \in  U^N\) for a time horizon $N \in  \N$.
% A third example is when the controller has the form of a feedforward neural network \cite{miller1995neural}, given by \[ \pi_\alpha(x)
%  = W_L \sigma_{L-1}\big(W_{L-1} \sigma_{L-2}\big(\cdots \sigma_1(W_1 x + b_1\big)\cdots\big) + b_L,
% \]
% where $L$ is the number of layers, \( W_i \in \mathbb{R}^{d_i \times d_{i-1}} \) represents the weight matrix for layer \( i \) with \( d_0 = n \), $d_L = m$ and \( b_i \in \mathbb{R}^{d_i} \) is the bias vector for layer \( i \) and  \( \sigma_i \) is the activation function for layer \( i \) (e.g., $\ReLU$, $\tanh$) \PJ{element-wise?} \textcolor{red}{yes}.
%  In this case, the set of parameters $\alpha \in \mathbb{R}^p$ represents the weights and biases of the neural network with dimension equal to $ \sum_{i = 0}^L( d_i \times d_{i-1} + d_i)$.

\subsection{Temporal specification}
\label{Sec:specfication}
Consider the system $\Sigma$ in~\eqref{eqn:non linear controlled system}. A specification $\psi \subset X^{N+1}$ is a set of admissible state sequences that defines the desired behavior of the system over a bounded time horizon $N \in \N$. This class of specifications is quite rich, and can cover specifications such as safety, reachability, and more complex linear temporal logic specifications over finite traces, LTL$_f$ \cite{LTLf17}. For example, an exact-time reachability at time $M \in \{0,1,\ldots,N\}$ of a set $A \subseteq X$, which is written in LTL$_f$ as $\psi = \nex^MA$, can be formulated as $$\psi = X^M\times A \times  X^{N-M} \subseteq X^{N+1}.$$
Similarly, finite-horizon reachability of a set $A \subseteq X$ between time instances $M_1$ and $M_2$, with $0 \leq M_1 < M_2 \leq N$, which is denoted by $\psi =\lozenge^{[M_1,M_2]} A$, can be formulated as
$$\psi= \bigcup\limits_{i=M_1}^{M_2}X^{i}\times A \times  X^{N-{i}} \subseteq X^{N+1}. $$
Finally, the finite-horizon safety of a set $A \subseteq X$ between time instances $M_1$ and $M_2$, with $0 \leq M_1 < M_2 \leq N$, which is denoted by $\psi =\square^{[M_1,M_2]} A$, can be formulated as
$$\psi= \bigcap\limits_{i=M_1}^{M_2}X^{i}\times A \times  X^{N-{i}} \subseteq X^{N+1}. $$
For the system $\Sigma$ in (\ref{eqn:non linear controlled system}), $ \xi( x, \pi_\alpha,\mathbf{d})$ is the set of state-input trajectories of the closed-loop system, over a bounded time horizon of length $N \in \N$, starting from a state $ x \in X$ under the feedback controller $\pi_\alpha$ and any disturbance input $ \mathbf{d} = (d(0),\ldots, d(N-1)) \in D^N$, defined formally as follows:
\begin{equation*}
\begin{aligned}
\label{trajectories set}
    \xi(x, \pi_\alpha, D&) =  \{\big((x(0), u(0)), (x(1),u(1)), \dots, (x(N)\big) \mid x(0) = x,\\ &x(k+1) = f(k, x(k) , u(k)) + d(k), u(k) = \pi_\alpha(x(k)), \\&  
    \text{ for all } d(k)   \in D   \text{ with } k \in \{0, \dots, N-1\}\}.
\end{aligned}
\end{equation*}
Additionally, we denote the projection of the trajectories on the state space as $ \xi_x(x, \pi_\alpha, D) \subseteq X^{N+1}$ and the projection on the input space as $\xi_u(x,\pi_\alpha, D) \subseteq U^{N}$ for $N \in \N $. The system $\Sigma$ in (\ref{eqn:non linear controlled system}) starting from $ x \in X$ under the feedback controller $\pi_\alpha$ is said to satisfy the specification $\psi \subseteq X^{N+1}$ if $\xi_x( x,\pi_\alpha,D) \subseteq \psi$.

In the rest of the paper, we focus on disturbance sets and input sets defined by balls centered at zero with respect to the infinity norm, denoted as \( D := \Omega_{\mu}(0) \) and $U:= \Omega_{\varepsilon}(0)$ respectively. Note that if the center of the ball is not zero, we can change the dynamics $f$ to make it zero. For simplicity, we use the shorthand notation 
\[
\xi(x, \pi_\alpha, \mu) := \xi(x, \pi_\alpha, \Omega_{\mu}(0)),
\]
and similar notation for \( \xi_x \) and \( \xi_u \). All of the preceding notation readily generalizes to open-loop control sequences by replacing a feedback policy $\pi_\alpha$ with a sequence of inputs $\uu = (u_0, \dots, u_{N-1}) \in  U^N$. For example, we write $\xi(x, \uu, \mu) $ to denote the set of all state-input trajectories, over the horizon $N \in \N$, starting from $ x \in X$ under the open-loop controller $\uu$ and disturbances $ \mathbf{d} \in \Omega_{\mu}(0)^N$.
\subsection{Resilience metric}
In this section, we define the resilience metric for a system in~\eqref{eqn:non linear controlled system} with the set of disturbances and inputs given by balls centered at zero: $D := \Omega_{\mu}(0)$ and $U:= \Omega_{\varepsilon_0}(0)$. This definition extends the concept originally introduced for autonomous systems in \cite{saoud2024, ait2025maximal}, which does not take into account input constraints.
\begin{definition}[Resilience Metric]
\label{def:controlled_resilience}
Consider the system $\Sigma$ in~\eqref{eqn:non linear controlled system}, a specification $\psi \subseteq X^{N+1}$, a controller $\pi_\alpha$ described in Section~\ref{system definition} and a point $x \in X$. We define the resilience metric  $g_{\psi}:X \times \mathbb{R}_{\ge 0} \rightarrow \mathbb{R}_{\ge 0}\cup\{+\infty\}$ of the system $\Sigma$ with respect to the initial state $x$ and the specification $\psi$ under the closed-loop controller  \(\pi_\alpha(x): X \rightarrow  \Omega_{\varepsilon_0} (0)\) for $\varepsilon_0 \geq 0$ as
\begin{equation}
\label{eq:resilience_closed}\hspace{-0.2em} 
g_\psi(x, \varepsilon_0): = 
\begin{cases}
\begin{aligned}
&\sup\big\{\mu\ge 0\,|\exists \alpha \in \R^d,\text{s.t} \\
& \hspace{2cm}\,\xi(x, \pi_\alpha,\mu)\subseteq \psi  \times  \Omega_{\varepsilon_0} (0)^N  \big\}, \\ 
& \hspace{1cm} \text{ if } \exists \alpha \in \R^d , \xi(x, \pi_\alpha,0) \in\psi\times  \Omega_{\varepsilon_0} (0)^N,\\
&0,  \hspace{0.7cm} \text{ if } \forall \alpha \in \R^d , \xi( x, \pi_\alpha,0)\notin \psi\times  \Omega_{\varepsilon_0} (0)^N.
\end{aligned}
\end{cases}
\end{equation}
For the open-loop controller, the definition of the resilience metric $g_{\psi}$ under the input sequence $\uu = (u_0, u_1, \dots,u_{N-1})\in \Omega_{\varepsilon_0} (0)^N$ for $\varepsilon_0 \geq 0$ becomes
\begin{equation}
\label{eq:resilience_open}\hspace{-0.2em} 
g_\psi(x, \varepsilon_0) = 
\begin{cases}
\begin{aligned}
&\sup\big\{\mu\ge 0\,|\exists\uu \in \Omega_{\varepsilon_0} (0)^N,\, \\& \hspace{2.5cm} \xi_x(x, \uu,\mu)\subseteq \psi    \big\}, \\ 
& \hspace{0.6cm} \text{ if } \exists \uu \in \Omega_{\varepsilon_0} (0)^N, \xi_x(x, \uu,0) \in\psi,\\
&0,  \hspace{0.3cm} \text{ if } \forall \uu \in \Omega_{\varepsilon_0} (0)^ N , \xi_x( x, \uu,0)\notin \psi .
\end{aligned}
\end{cases}
\end{equation}
We denote $g_\psi(x)$ the resilience metric in the absence of input constraints, i.e., $g_\psi(x) = g_\psi(x, 0)$.
\end{definition}

% The resilience metric can also be defined for a set $ A \subseteq X $ and the specification $\psi$ as follows: 
% \begin{equation}
% \hspace{-0.2em} 
% g_\psi(A) = 
% \begin{cases}
% \begin{aligned}
% &\sup_{\mu \geq 0, \alpha \in \R^d}\big\{\mu\ge 0\,|\,\xi_x(A, \pi_\alpha,\mu)\subseteq \psi \big\}, \\ 
% & \hspace{2.2cm}  \text{if }\forall x \in A, \exists \alpha \in  \R^d,  \xi(x,\pi_\alpha, 0)\in\psi\\
% &0,  \hspace{1.5cm} \text{if }\exists x \in A ,\forall \alpha \in A\times \R^d, \xi_x(x, \pi_\alpha, 0)\not\in\psi
% \end{aligned}
% \end{cases}
% \end{equation}

% % how we can update this definition in the negative case to add something on it.

This definition formulates the resilience metric \( g_\psi \) that evaluates for a given $x \in X$ the maximum disturbance bound \(\mu\) and the optimal parameter \(\alpha\) ensuring the trajectories in \(\xi_x(x, \pi_\alpha,\mu)\) satisfy the specification \(\psi\) while respecting input constraints \( \Omega_{\varepsilon_0}(0)^N \). The notation distinguishes between set inclusion (\(\subseteq\)) for disturbed trajectories and element membership (\(\in\)) for nominal cases because \(\xi(x, \pi_\alpha,\mu)\) represents all possible state-input trajectories under any disturbance sequence \(\dd = (d(0),\dots,d(N-1)) \in \Omega_\mu(0)^{N}\), while \(\xi_x(x, \pi_\alpha,0)\) refers to one single nominal state trajectory without disturbances.
The case of $g_\psi(x, \varepsilon_0)$ equal to zero corresponds to the case where there is no controller that can lead the nominal state trajectory \(\xi_x(x, \pi_\alpha,0)\) to satisfy the specification $\psi$. The same reasoning applies to open-loop controllers by replacing the policy $\pi_\alpha$ by a sequence of inputs $\uu \in \Omega_{\varepsilon_0}(0)^N$.
\subsection{Effort metric}
In the next definition, we introduce the effort metric, which measures the minimum input magnitude required to drive a discrete-time control system, under a predefined level of disturbance given by a ball centered at zero, $D := \Omega_{\mu_0}(0)$, to satisfy a given specification over a finite horizon.

\begin{definition}[Effort Metric]
\label{def:effort_metric_unified}
Consider the system \( \Sigma \) in~\eqref{eqn:non linear controlled system}, a specification $\psi \subseteq X^{N+1}$, a controller $\pi_\alpha$ as in Section~\ref{system definition}, an initial state  \( x \in X \), and a disturbance bound $\mu_0 \geq0$. We define the effort metric \( h_\psi : X\times \R_{\geq 0} \to \mathbb{R}_{\geq 0} \cup \{+\infty\} \) of the system $\Sigma$ with respect to the initial state $x$, the specification $\psi$, and the closed-loop controller \(\pi_\alpha(x): X \rightarrow U\) as
\begin{equation}
\label{def:effort_closed}\hspace{-0.2em} 
h_\psi(x, \mu_0): = 
\begin{cases}
\begin{aligned}
& \inf\big\{\varepsilon \geq 0\,\mid \exists\alpha \in \R^d \, \text{s.t.} \\ 
& \hspace{1.45cm}  \xi(x, \pi_\alpha, \mu_0) \subseteq \psi \times\Omega_{\varepsilon}(0)^N \big\}, \\
& \hspace{0.9cm} \text{ if } \exists \alpha \in \R^d , \xi_x(x, \pi_\alpha, 0) \in\psi,\\
&0,  \hspace{0.6cm} \text{ if } \forall \alpha \in \R^d , \xi_x( x, \pi_\alpha, 0)\notin \psi.
\end{aligned}
\end{cases}
\end{equation}

For the open-loop controller $\uu \in U^N$, the definition of the effort metric is
\begin{equation}
\label{def:effort_open}\hspace{-0.2em} 
    h_\psi(x, \mu_0) := 
    \begin{cases}
\begin{aligned}
&\inf\big\{ \varepsilon \ge 0 \,\mid  \exists \uu \in \R^{m\times N}, \,  \text{s.t.} \\
& \hspace{1.4cm} \xi(x, \mathbf{u}, \mu_0) \subseteq \psi\times \Omega_{\varepsilon}(0)^N \big\},\\
 & \hspace{0.8cm} \text{ if } \exists \uu \in \Omega_{\varepsilon}(0)^N , \xi_x(x, \uu,0) \in\psi,\\
&0,  \hspace{0.5cm} \text{ if } \forall \uu \in \Omega_{\varepsilon}(0)^N , \xi_x( x, \uu,0)\notin \psi.
\end{aligned}
\end{cases}
\end{equation}
We denote $h_\psi(x)$ the effort metric for the nominal dynamical system without perturbations, i.e., $h_\psi(x) = h_\psi(x, 0)$.
\end{definition}
\subsection{Resilience-effort metric}
In the next definition, we introduce the resilience-effort metric, which measures the trade-off between maximal resilience and minimal effort required to drive a system, using predefined weights $w_1, w_2 \ge 0$ for each metric, to satisfy a given specification over a finite horizon.
\begin{definition}[Resilience-Effort Metric]
\label{def:r_e_metric_unified}
Consider system \( \Sigma \) in~\eqref{eqn:non linear controlled system}, a specification $\psi \subseteq X^{N+1}$ and a controller $\pi_\alpha$ as in Section~\ref{system definition}, an initial state \( x \in X \) and positive weights $w_1, w_2 \geq0$. We define the resilience-effort metric \( p_\psi : X\times \R_{\geq 0} \to \mathbb{R}_{\geq 0} \cup \{+\infty\} \) of the system $\Sigma$ with respect to the initial state $x$, and the specification $\psi$ under the closed-loop controller \(\pi_\alpha(x): X \rightarrow U\) as
\begin{equation*}
% \label{def:r_e_closed}\hspace{-0.2em} 
p_\psi(x, w_1, w_2): = 
\begin{cases}
\begin{aligned}
& \sup_{} \big\{ w_1  \mu - w_2\varepsilon\,\mid \mu \ge 0, \varepsilon \geq 0, \exists \alpha \in \R^d   \\ & \hspace{1.7cm} \text{ s.t.} \; \xi(x, \pi_\alpha, \mu)  \subseteq \psi \times\Omega_{\varepsilon}(0)^N \big\}, \\
& \hspace{1.9cm} \text{ if } \exists \alpha \in \R^d , \xi_x(x, \pi_\alpha, 0) \in\psi,\\
&0,  \hspace{1.6cm} \text{ if } \forall \alpha \in \R^d , \xi_x( x, \pi_\alpha, 0)\notin \psi.
\end{aligned}
\end{cases}
\end{equation*}

For the open-loop controller $\uu \in U^N$, the definition of the resilience-effort metric is
\begin{equation*}
\label{def:r_e_open}\hspace{-0.2em} 
    p_\psi(x, w_1, w_2) := 
    \begin{cases}
\begin{aligned}
&\sup\big\{ w_1  \mu - w_2\varepsilon \mid \mu \geq 0, \varepsilon \geq 0, \exists \uu \in \R^{m\times N}   \\ & \hspace{1.6cm}\text{s.t.} \; \xi(x, \mathbf{u}, \mu)   \subseteq \psi\times \Omega_{\varepsilon}(0)^N \big\},\\
 & \hspace{1.5cm} \text{ if } \exists \uu \in \Omega_{\varepsilon}(0)^N , \xi_x(x, \uu,0) \in\psi,\\
&0,  \hspace{1.2cm} \text{ if } \forall \uu \in \Omega_{\varepsilon}(0)^N , \xi_x( x, \uu,0)\notin \psi.
\end{aligned}
\end{cases}
\end{equation*}
\end{definition}
\begin{remark}
\label{substitution min_inf}
Since we assumed $f$ is a continuous function, the closed-loop system follows continuous dynamics. Hence, when considering closed specifications, the supremum/infimum operator in the definition of resilience/effort can be replaced by the maximum/minimum operator \cite{saoud2024}. 
% This substitution will be adopted in the rest of this paper when dealing with closed spaecification.  \textcolor{red}{I will add the fact that we use a closed specification for the linear system part. Linear and not linear is not defined yet here.}
\end{remark}

\subsection{Problem formulation}
Consider the system $\Sigma$ in~\eqref{eqn:non linear controlled system}, a specification $\psi \subseteq X^{N+1}$, a controller template $\pi_\alpha$ as in Section~\ref{system definition} and an initial state $x \in X$. The objective is to design the optimal controller $\pi_\alpha$ such that all the trajectories satisfy the specification $\psi$, i.e., $\xi_x(x, \pi_\alpha,\mu) \subseteq \psi$ taking into account input constraints for the input trajectories $\xi_u(x, \pi_\alpha,\mu) \subseteq \Omega_{\varepsilon_0}(0)^N$ for a given bound $\varepsilon_0 \geq 0$, while also maximizing the resilience metric $g_{\psi}(x, \varepsilon_0)$. For the effort metric, the goal is to find the optimal controller that drives all trajectories, with a given bound $\mu_0$ on the disturbance set, to satisfy the specification $\psi$, i.e., $\xi_x(x, \pi_\alpha,\mu_0) \subseteq \psi$ while minimizing the effort metric $h_\psi(x, \mu_0)$ which corresponds to the bound on the input. The same development is applied to the open-loop system considering an open-loop input sequence $\uu=(u_0,u_1,\dots,u_{N-1})$ as defined in Section~\ref{Sec:specfication} instead of a policy $\pi_\alpha$.

In the following section, we provide a closed-form expression for computing the resilience–effort metric. We then present corollaries for the particular cases of the resilience and effort metrics and analyze the computational complexity.

%\textcolor{blue}{This goal is translated to a robust optimization problem defined by the resilience metric $g_\psi$ ( see Definition~\ref{def:controlled_resilience}). In the first case, we consider a linear dynamical system and a linear controller, aiming to find the exact solution for the problem. In the second case, we work with a nonlinear system, a nonlinear controller, and the specification $\psi$, seeking to synthesize the optimal controller under maximum disturbance. We leverage techniques from scenario optimization to ensure probabilistic guarantees for the solution.}

\section{Trade-off between resilience and effort}
\label{sec:3}
In this section, we address the problem of computing the resilience-effort metric for a weighted multi-objective cost function, along with the controller that achieves this optimal trade-off. This problem falls within the framework of multi-objective optimization (or Pareto optimization) \cite{chinchuluun2008pareto}, where the primary goal is to characterize the trade-off point between effort and resilience for a given set of preference weights. We will first present exact results for linear systems, then extend to nonlinear systems.

% Following the methodology of the previous section, we first present an exact computational solution for the case of time-varying linear systems with linear controllers. Subsequently, we extend our approach to nonlinear systems and nonlinear controllers using scenario optimization.
\subsection{Time-varying linear systems and closed-loop linear controllers}
Consider the system $\Sigma$ in~\eqref{eqn:non linear controlled system} with a linear time-varying (LTV) dynamics
\begin{equation}
\begin{aligned}
x(k+1) 
&= A_kx(k) + B_ku(k) + d(k), 
\end{aligned}
\label{eqn:time varying linear sys}
\end{equation}
where $A_k \in  \R^{n \times n}$ and $B_k \in \R^{n \times m}$ are the state and input matrices of the dynamics at time $k \geq 0$. In this section, we consider a closed-loop controller defined by \( \pi_\alpha(x(k)) =  u(k) = \alpha_1 x(k) + \alpha_2 \), where $(\alpha_1, \alpha_2) \in \mathbb{R}^{m\times n} \times \mathbb{R}^{m}$. We focus on specifications defined by \begin{equation}
\label{polytopic_spec}
\psi = \Gamma_0 \times \Gamma_1 \times \dots \times \Gamma_N \subseteq X^{N+1},
\end{equation}
where for $k=0,1,\ldots,N$, and $\Gamma_k = \{x \mid G_kx \leq H_k\}$ is a polytopic set with $G_k \in \mathbb{R}^{q \times n}$ and $H_k \in \mathbb{R}^q$. Note that such specifications include exact time reachability, finite-horizon safety and the more general class of closed and convex LTL$_f$ specifications (for details on convex and closed specifications, please refer to Definitions 3.2 and 3.3, as well as Section 7.3 in \cite{saoud2024}). The substitution of the supremum/infimum operator will be adopted whenever we use this polytopic specification as per Remark~\ref{substitution min_inf}. The next theorem is one of the main contributions of this paper, we show that the computation of the trade-off metric with input constraints can be translated into a standard optimization problem.

\begin{theorem}
\label{Theorem:7}
Consider the LTV system $\Sigma$ in~\eqref{eqn:time varying linear sys} with a closed-loop controller $\pi_\alpha(x(k)) = \alpha_1 x(k) + \alpha_2$. Consider the specification $\psi$ as defined in~\eqref{polytopic_spec}. For given weights $w_1, w_2 \geq 0$, the Pareto-optimal trade-off between resilience and control effort is characterized by
\begin{equation}
\label{eq:pareto_optimization}
\begin{aligned}
p_\psi(x, w_1, w_2) = &\max_{\substack{\varepsilon \ge 0,\; \mu \ge 0,\; \alpha_1 \in \mathbb{R}^{m \times n},\\ \alpha_2 \in \mathbb{R}^{m},\; P \ge 0}}  \; w_1  \mu - w_2\varepsilon \\[4pt]
\text{s.t.} \quad & P A_b = \mu \,E^{cl}(\alpha_1), \\[2pt]
&  P B_b \le F^{cl}(x, \alpha_1, \alpha_2, \varepsilon), \\[2pt]
& P \ge 0,
\end{aligned}
\end{equation}
where matrices $A_b, B_b, E^{cl}(\cdot), F^{cl}(\cdot)$ are defined below as
\begin{equation}
\label{Ab_Bb_closed_1}
A_b = \begin{bmatrix} \mathbb{I}_{n(N+1)} \\ -\mathbb{I}_{n(N+1)} \end{bmatrix} , \;
B_b = \mathbf{1}_{2n(N+1)} ,
\end{equation}
\begin{equation}
\label{E_F_closed_1}
E^{cl}(\alpha_1) = \begin{bmatrix} E_x(\alpha_1) \\ E_u(\alpha_1) \end{bmatrix}, \quad
F^{cl}(x, \alpha_1, \alpha_2, \varepsilon) = \begin{bmatrix} F_x(x, \alpha_1, \alpha_2) \\ F_u(x, \alpha_1, \alpha_2, \varepsilon) \end{bmatrix}.
\end{equation}
The matrix blocks are specified as follows:
\begin{itemize}
    \item For the state-dependent part
\end{itemize}
\begin{equation}
\begin{aligned}
\label{Ex_Fx_closed_1}
E_x (\alpha_1)= \begin{bmatrix} E_{x0} \\ E_{x1} \\ \vdots \\ E_{xN} \end{bmatrix}, \;
F_x(x,\alpha_1, \alpha_2)= \begin{bmatrix} F_{x0} \\ F_{x1} \\ \vdots \\ F_{xN} \end{bmatrix} ,
% \in \mathbb{R}^{q \times n(N+1)}, 
\\
E_{xk} = G_k E_k, \quad
F_{xk} = H_k - G_k F_{k}(x, \alpha_1, \alpha_2),
\end{aligned}
\end{equation}
with $E_k = \begin{bmatrix} 0_{n},\underbrace{\underline{A}_{k,0} , \underline{A}_{k,1},  \cdots\underline{A}_{k,k-1}}_{k \, \text{ times}} , \underbrace{0_{ n} , \cdots , 0_{ n}}_{N-k \, \text{ times}} \end{bmatrix}$, $ F_{k}(x, \alpha_1, \alpha_2) =  \underline{A}_{k,-1} x +  \sum_{i=0}^{k-1}  \underline{A}_{k,i}  B_i\alpha_2$, $\underline{A}_{k,i} = \prod_{j=i+1}^{k-1}\bar{A}_j,$ $\bar{A}_k = A_k + B_k\alpha_1,$ and the convention that an empty product equals the identity, i.e., $ \underline{A}_{k,k-1}=A_{1,0}=I_n$ and $\underline{A}_{0,i}=0_n$.
\begin{itemize}
 \item For the input-dependent part
\end{itemize}
\begin{equation}
\label{E_u, F_u closed_1}
\begin{aligned}
E_u(\alpha_1) = \begin{bmatrix} E_{u0} \\ E_{u1} \\ \vdots \\ E_{u(N-1)} \end{bmatrix}, \hspace{0.1cm}
F_u(x, \alpha_1, \alpha_2, \varepsilon) = \begin{bmatrix} F_{u0} \\  F_{u1} \\ \vdots \\ F_{u(N-1)} \end{bmatrix}, \\
E_{uk} = A_u \alpha_1 E_k, \; 
F_{uk} = \varepsilon B_u - A_u S_k(x, \alpha_1, \alpha_2) \in \mathbb{R}^{2m},
\end{aligned}
\end{equation}
with the convention $E_{u0} = 0_{2m \times n(N+1)}$.
Finally, the auxiliary matrices $A_u$, $B_u$ and $S_k$ are given by
\begin{equation}
\begin{aligned}  
\label{A_u, b_u closed_1}
&A_u = \begin{bmatrix} \mathbb{I}_m \\ -\mathbb{I}_m \end{bmatrix} \in \mathbb{R}^{2m \times m}, \quad
B_u = \mathbf{1}_{2m} \in \mathbb{R}^{2m},\\
\end{aligned}
\end{equation}
\begin{equation}
\begin{aligned}  
\label{s_k}
&S_k(x, \alpha_1, \alpha_2) = \alpha_1 F_{k}(x,\alpha_1, \alpha_2)  + \alpha_2 \in \R^m.
\end{aligned}
\end{equation}
\end{theorem}
\begin{proof}[Proof Sketch]
The intuition behind the proof is to consolidate all state and input constraints over the finite horizon into a single, unified matrix inequality of the form $E^{cl} Y \leq F^{cl}$, where $Y = [d_0, \dots, d_{N-1}]^\top$ is the stacked disturbance vector. Once in this form, the robust constraint ``$E^{cl} Y \leq F^{cl}$ for all $Y \in \Omega_{\mu}^N(0)$'' can be transformed via Farkas' lemma~\cite{Alexander1999} (see Lemma~\ref{theorem farkas} in Appendix) into a tractable program. The full proof can be found in the appendix.
% The proof proceeds in three systematic steps. First, using the system dynamics, we express all state constraints $G_k x(k) \leq H_k$, $k = 0,\dots,N$, as a single matrix inequality \(E_x Y \leq F_x,\) where $E_x$ and $F_x$ depend only on the initial state, the system matrices, and the controller parameters. Similarly, we express each input constraint $u(k) \in \Omega_{\varepsilon}(0)$ (or, in feedback form, $\alpha_1 x(k) + \alpha_2 \in \Omega_{\varepsilon}(0)$) as a matrix inequality. By stacking these inequalities over $k = 0,\dots,N-1$, we obtain
% \(E_u Y \leq F_u,\)
% where $E_u$ and $F_u$ capture the dependence on the controller and the effort bound $\varepsilon$. Vertically concatenating the state and input constraints yields the compact form
% \[
%     E^{cl} Y \leq F^{cl}, \qquad E^{cl} = \begin{bmatrix} E_x \\ E_u \end{bmatrix}, \; F^{cl} = \begin{bmatrix} F_x \\ F_u \end{bmatrix}.
%     \]
% The requirement that this inequality hold for every $Y \in \Omega_{\mu}^N(0)$ is then equivalent, via Farkas' lemma, to the existence of a nonnegative matrix $P$ satisfying 
% \(P A_b = \mu E^{cl}, \qquad P B_b \leq F^{cl},\)
% where $A_b$ and $B_b$ are constant matrices that encode the geometry of the disturbance set $\Omega_{\mu}(0)$. This stacked form enables the application of Farkas' lemma, thus transforming robust feasibility constraints, which involve infinitely many constraints, into finite feasibility constraints. 
\end{proof}

\subsection{Time-varying linear systems and open-loop controllers}
The next theorem provides the approach for computing the resilience–effort trade-off for LTV system $\Sigma$ in~\eqref{eqn:time varying linear sys}, polytopic specifications in~\eqref{polytopic_spec}, and open-loop controllers. 

\begin{theorem}
\label{Theorem:8}
Consider the linear system $\Sigma$ in~\eqref{eqn:time varying linear sys} with an open-loop controller defined by a sequence of inputs $\uu = (u_0,u_1,\dots, u_{N-1}) \in U^N$ and the specification $\psi$ as in~\eqref{polytopic_spec}. For given weights $w_1, w_2 \geq 0$, the Pareto-optimal trade-off between resilience and control effort is characterized by
\begin{equation}
\label{eq:pareto_open_loop}
\begin{aligned}
p_\psi(x, w_1, w_2) = &\max_{\substack{\mu \ge 0,\; \varepsilon \ge 0,\; \\ \mathbf{\tilde{u}} \in \Omega_{1}(0)^N,\; P \geq 0}} \; w_1\mu - w_2 \varepsilon \\
\text{s.t.} & \quad P A_b = \mu  \,E^{ol} ,\\
& \quad P B_b \leq F^{ol}(x, \varepsilon, \mathbf{\tilde{u}}), \\
& \quad P \geq 0,
\end{aligned}
\end{equation}
where the matrices $A_b$ and $B_b$ are defined in Theorem~\ref{Theorem:7}, $\uu = \varepsilon \tilde{\uu}$ and the matrices $E^{ol}$ and $F^{ol}(.)$ are defined below.
\begin{equation}
\begin{aligned}
\label{E_F_open_loop_1}
&E^{ol} = \begin{bmatrix} E_{0} \\  E_{1} \\\vdots \\ E_{N} \end{bmatrix}, \hspace{0.1cm}
 F^{ol}(x, \uuu, \varepsilon) = \begin{bmatrix} F_{0} \\ F_{1} \\ \vdots \\ F_{N} \end{bmatrix},  
% \in \mathbb{R}^{qN \times nN},
\\
&
F_{k} = H_k - G_k  \underline{A}_{k,-1} x -\varepsilon G_k  \sum_{i=0}^{k-1}\underline{A}_{k,i} B_i\tilde{u}(i),
% \in \R^{q},
\end{aligned}
\end{equation}
Where $E_k = G_k \begin{bmatrix} 0_{n} ,\underbrace{ \tilde{A}_{k,0} ,  \cdots ,\tilde{A}_{k,k-1}}_{k \text{ times}}, \underbrace{0_{n} , \cdots , 0_{n}}_{N -k \text{ times}} \end{bmatrix}$, $\tilde{A}_{k,i} = \prod_{j=i+1}^{k-1}A_j$, with the convention that an empty product equals the identity, i.e., $ \tilde{A}_{k,k-1}=A_{1,0}=I_n$ and $\tilde{A}_{0,i}=0_n$.
 % $P \in \R^{q(N+1) \times 2n(N+1)}$
\end{theorem}
\begin{proof}[Proof Sketch]
The proof follows similar steps as that of Theorem~\ref{theorem:1} and can be found in the appendix.
% The intuition is to consolidate all state constraints over the finite horizon into a single matrix inequality of the form \(E^{\text{ol}} Y \leq \tfrac{1}{\mu}F^{\text{ol}}\), where \(Y = \tfrac{1}{\mu}[0, d_0, \dots, d_{N-1}]^\top\) is the scaled disturbance. This allows the robust constraint ``\(E^{\text{ol}} Y \leq \tfrac{1}{\mu}F^{\text{ol}}\) for all \(Y\) such that \(A_b Y \leq B_b\)'' to be reformulated via Farkas' lemma.

% The proof proceeds in three steps. First, using the state evolution, all constraints \(G_k x(k) \leq H_k\), \(k=0,\dots,N\), are written as
% \begin{equation}
% E^{\text{ol}} Y \leq \tfrac{1}{\mu}F^{\text{ol}}(x, \varepsilon, \tilde{\mathbf{u}}).
% \end{equation}
% Second, \(A_b Y \le B_b\) encodes \(\|Y\|_\infty \le 1\). Multiplying by \(\mu\) gives
% \begin{equation}
% \mu E^{\text{ol}} Y \leq F^{\text{ol}}(x, \varepsilon, \tilde{\mathbf{u}}), \quad \forall Y:\; A_b Y \leq B_b.
% \end{equation}
% By the affine Farkas lemma, this holds iff there exists \(P \ge 0\) such that
% \begin{equation}
% P A_b = \mu E^{\text{ol}}, \qquad P B_b \leq F^{\text{ol}}(x, \varepsilon, \tilde{\mathbf{u}}),
% \end{equation}
% which leads directly to \eqref{eq:pareto_open_loop}.
\end{proof}

In this subsection, we transformed the original uncertain problem of the Pareto optimization into a tractable optimization problem even when both $\mu$ and $\varepsilon$ are decision variables. 
For the linear closed-loop controller, the constraints of the optimization problem are linear in $\varepsilon$ and polynomial in $\alpha_1$ and $\alpha_2$ and $\mu$. For the open-loop controller the constraints of the optimization problem are linear in $\mu$ and bilinear in $\varepsilon$ and $\uuu$. In the next subsection, we extend the computation of the resilience-effort metric to the case of nonlinear systems.

\subsection{Nonlinear systems and nonlinear controllers}
\label{scenario_optimization_steps}
This section presents a scenario-based approach to compute the resilience-effort metric for nonlinear systems, nonlinear closed-loop controllers, and general specifications. With this approach, we can handle state constraints defined by the specification, as well as input constraints. Although the scenario approach provides only an approximate solution to the robust optimization problem, it ensures computational feasibility for practical implementation. Moreover, it also allows providing probabilistic guarantees on the obtained result \cite{garatti2024}. We first formulate the calculation of the resilience metric as a robust optimization problem and then translate it to a scenario optimization problem using disturbance normalization step.
 
\subsubsection{Robust Optimization}
Consider the system defined in~\eqref{eqn:non linear controlled system}, with a constrained input set given by $U  \subseteq \R^m$. To properly state the scenario optimization problem, let us define the robust optimization problem for a given $x \in X$, a specification $\psi \subseteq X^{N+1}$, and weights \(w_1, w_2 \geq 0\) as
\begin{equation}
\label{eqn:scenario_pb_rew_ro_re}
\begin{aligned}
\sup_{(\mu, \varepsilon, \alpha) \in \mathbb{R}_{\geq 0} \times \mathbb{R}_{\geq 0} \times \mathbb{R}^d} \; & w_1 \mu - w_2 \varepsilon \\
\text{s.t. } \; & (\mu, \varepsilon, \alpha)  \in \bigcap_{\mathbf{d} \in \Omega_{\mu}(0)^N}  \mathcal{X}_{\mathbf{d}},
\end{aligned}
\end{equation}
where the constraint set $\mathcal{X}_{\mathbf{d}}$ is given by
\begin{equation}
\label{eqn:cont_constr_trade}
\begin{aligned}
\mathcal{X}_{\mathbf{d}} = \big\{ &(\mu,\varepsilon, \alpha) \in \mathbb{R}_{\ge 0} \times \mathbb{R}_{\ge 0} \times\R^d  \mid  \nonumber  \mathbf{d} = (d_0, \ldots, d_{N-1}) \in \Omega_{\mu}(0)^N \nonumber\\
& \xi(x,\pi_\alpha, \mu) \subseteq \psi \times U^N   \big\}.
\end{aligned}
\end{equation}
One way to solve this problem is to consider all the possible values of $\mathbf{d} = (d_0, d_1, \dots, d_N ) \in \Omega_{\mu}(0)^N$ and look for a solution. This is not possible because of the infinite realization possibilities of $\mathbf{d}$. To address this, one approach is to randomly sample $\mathbf{d}$ and solve the problem for the samples. This approach is called the scenario approach \cite{garatti2024}.
\subsubsection{Scenario Optimization}

To formulate the scenario optimization problem, we define the normalized disturbance set 
$\mathcal{D} = \{ \boldsymbol{\delta} = (\delta_0, \dots, \delta_{N-1}) \in \Omega_1(0)^N\}$, where $\mathbf{d} = \mu \boldsymbol{\delta}$ represents the actual disturbance in 
$\Omega_\mu(0)^N$. 

% For a given $N \in \mathbb{N}$, the state evolution under a fixed $\boldsymbol{\delta} \in \mathcal{D}$ follows the recursive relation $x_{\boldsymbol{\delta}}(N) = f\big(f(\dots f(x, \pi_\alpha(x)) + \mu \delta_0, \dots\big) + \mu \delta_{N-1}.$
% The corresponding trajectory $\xi(x, \pi_\alpha, \mu \boldsymbol{\delta})$ is the sequence of states and controls given by
% $$
% \begin{aligned}
% \label{trajectory}
%     &\xi(x, \pi_\alpha, \mu \boldsymbol{\delta}) = \big\{\big(x_{\boldsymbol{\delta}}(0), u(0)\big), \big(x_{\boldsymbol{\delta}}(1),u(1)\big), \ldots, x_{\boldsymbol{\delta}}(N))  \mid \\
%     &x(0) = x, u(k) = \pi_\alpha(x_{\boldsymbol{\delta}}(k)),
%     x_{\boldsymbol{\delta}}(k+1) = f(x_\delta(k) , u(k)) +\mu\delta_k  \big\}.
% \end{aligned}
% $$

For a state $x \in X$ and a disturbance trajectory $\boldsymbol{\delta} \in \mathcal{D}$, we define the set of constraints as follows:
\begin{align*}
\mathcal{X}_{\boldsymbol{\delta}} = \{ (\mu, \varepsilon, \alpha)\in \mathbb{R}_{\geq 0} \times \mathbb{R}_{\geq 0} \times\mathbb{R}^d \; \mid\\ 
\; \xi(x, \pi_\alpha, \mu \boldsymbol{\delta}) \subseteq \psi \times \Omega_\varepsilon(0) \}.
\end{align*}
This ensures the disturbed trajectory meets the specification $\psi$ while respecting input constraints $U$.

For \(i \in \{1, \dots, M\}\), we consider \(M\) scenarios, \(\boldsymbol{\delta}_i = (\delta_0^i, \dots, \delta_{N-1}^i)\), which are independent and identically distributed (i.i.d) from the probability space \((\mathcal{D}, \mathcal{F}, \mathbb{P})\), where $\mathcal{F}$ is the Borel $\sigma$-algebra on $\mathcal{D}$ and \(\mathbb{P}\) is any probability measure and \(M \in \mathbb{N}\). With these concepts in hand, the scenario optimization problem is formulated as
\begin{equation}
\label{eqn:scenario_pb_rew_re}
\begin{aligned}
\sup_{(\mu, \varepsilon, \alpha) \in \mathbb{R}_{\geq 0} \times \mathbb{R}_{\geq 0} \times \mathbb{R}^d} \; & w_1 \mu - w_2 \varepsilon \\
\text{s.t. } \; & (\mu, \varepsilon, \alpha) \in \bigcap_{i=1}^{M} \mathcal{X}_{\boldsymbol{\delta}_i},
\end{aligned}
\end{equation}
This nonconvex optimization problem seeks the optimal solution denoted by $\theta_M^* = (\mu_M^*, \varepsilon_M^*, \alpha_M^*) $ that is feasible for all $M$ scenarios and can be solved with known methods such as sequential quadratic programming or interior-point techniques \cite{NumericalOp}. Furthermore, we rely on the results in \cite{garatti2024} to provide a probabilistic guarantee for the generalization of the solution $\theta_M^*$ to unseen constraint scenarios. We will formulate this guarantee in the rest of this subsection. 

For a fixed number of scenarios $M$, we assume that a feasible and locally optimal solution $\theta_M^*  \in \mathbb{R}_{\ge 0} \times \mathbb{R}_{\ge 0} \times \mathbb{R}^d$ for problem~\eqref{eqn:scenario_pb_rew_re} is available and define the violation probability or the risk of the solution as $V(\theta_M^*)=\mathbb{P}\left\{\boldsymbol{\delta} \in \mathcal{D}: \theta_M^* \notin \mathcal{X}_\delta\right\}$. This metric measures the generalization power of the optimal solution to unseen constraints. 
{It provides the probability that a new constraint will not be satisfied by the solution $\theta_M^*$, indicating a violation of the considered solution. This indicates the probability with which the obtained solution can be generalized to the set of uncertainties $\mathcal{X}_{\mathbf{d}}$ for any ${\mathbf{d}}  \in \Omega_{\mu}(0)^N$. In the perfect case where $V(\theta_M^*) = 0$, the solution $\theta_M^*$ remains valid for any scenario constraint in the continuous space $\mathcal{X}_\mathbf{d}$ for any ${\mathbf{d}}  \in \Omega_{\mu}(0)^N$.
A constraint in the scenario program~\eqref{eqn:scenario_pb_rew_re} is called support constraint if its removal modifies the optimal solution $\theta_M^*$, and the complexity $s_M^*$ of $\theta_M^*$ is the number of such support constraints. For more details on the metric $s_M^*$, called complexity, please refer to \cite{RiskComplexity}.

% it is essential to ensure that the solver used delivers repeatable results which requires using deterministic solvers. 

We define the discrete function $b(k)$ for $k=$ $0,1, \ldots, M$ as
\begin{equation}
\label{eqn:epsilon}
b(k):=1-t(k), \quad k=0,1, \ldots, M-1, \text { and } b(M)=1,
\end{equation}
where $t(k)$ is the unique solution of the following polynomial equation for a chosen confidence parameter $\beta \in (0, 1]$:
\begin{equation}
\label{risk_bound}
\frac{\beta}{M} \sum_{m=k}^{M-1}\binom{m}{k} t^{m-k}-\binom{M}{k} t^{M-k}=0.
\end{equation}
This metric establishes the bound on the violation probability. We can now present the main theorem in this section.
\begin{theorem}
\label{theorem:9}
Consider the system $\Sigma$ in~\eqref{eqn:non linear controlled system}. For $x \in X$, a specification $\psi \subseteq X^{N+1}$ as defined in \ref{Sec:specfication}, and an input set $U \subseteq \R^m$, let $\theta^*_M$ be the solution to the optimization problem defined in~\eqref{eqn:scenario_pb_rew_re} for $M \in \mathbb{N}$ number of scenarios. For any probability measure $\mathbb{P}$, for any confidence $\beta \in (0, 1]$ and with $b(k)$, $ k = 0, 1, \dots, M$ as given in~\eqref{eqn:epsilon}, it holds that
\[
\mathbb{P}^M \big(V(\theta_M^*) \le b(s_M^*)\big) \ge 1-\beta
\]
% where \( s_M^* \) is called the complexity of the solution and refers to the number of scenarios \( \delta_i \) for \( i \in \{1, \dots, M\} \) that alter the solution \( \theta_M^* \) if removed.
\end{theorem}

\begin{proof}
To derive formal guarantees from solving problems in~\eqref{eqn:scenario_pb_rew_re}, we must satisfy a requirement formalized in Property $1$ in \cite{garatti2024}. This property has been proven to hold for robust optimization problems in \cite[Section 3]{garatti2024}. 
Using a deterministic solver with consistent initialization ensures that the optimization procedure is reproducible, in the sense that the same problem instance defined by an identical set of scenarios and initial conditions will always lead to the same numerical solution. Since the solver follows a fixed computational path and does not rely on randomness, repeated executions yield identical results.
Therefore, for the optimization problem defined in~\eqref{eqn:scenario_pb_rew_re}, this reproducibility property is satisfied. Thus, using \cite[Theorem 6]{garatti2024} for a robust optimization problem provides the desired result. 
\end{proof}

\medskip

The information on the number of scenarios M sufficient to achieve a level of
confidence is given implicitly by Theorem~\ref{theorem:9}. We begin by selecting a desired confidence level, \(\beta \in (0, 1]\), which can be made arbitrarily close to 1. For a chosen number of scenarios \(M\), we solve the associated optimization problem and determine the corresponding complexity \(s_M^*\). Using this value, we compute the bound on the risk probability, denoted by \(b(s_M^*, M)\), by solving Equation~\eqref{risk_bound}.
The theory of the scenario approach guarantees that for any \(\beta\), the probability that the true violation probability of the solution exceeds \(b(s_M^*, M)\) is at most \(\beta\). Note that the complexity \(s_M^*\) itself depends on \(M\). For a fixed complexity \(k\), the bound \(b(k, M)\) is a decreasing function of \(M\). In practice, \(k = s_M^*\) is typically stable as \(M\) grows. This observation, combined with the theoretical behavior of \(b(k, M)\), implies that increasing \(M\) leads to a tighter risk bound \(b(s_M^*, M)\), allowing us to make this bound as small as desired.
Hence, one may choose the higher value of $M$ until reaching the desired violation probability bound level. This shows that any confidence parameter and bound can be achieved by an appropriate choice of the number of scenarios $M$. 

\begin{remark}
A key advantage of this method is that it is model-free: it requires only measurements of the system dynamics at samples of the normalized perturbation. Specifically, for a given normalized perturbation $\boldsymbol{\delta}_i \in \mathcal{D}$, we only need to observe the resulting state $x(k)$ for each time step $k = 1, 2, \dots, N$, without the explicit knowledge of the dynamics. This information is sufficient to formulate the optimization problem in \eqref{eqn:scenario_pb_rew_re} and, consequently, to compute an approximate resilience metric with probabilistic guarantees. 
\end{remark}
In the next section, we provide the approach to compute the resilience metric in both linear and nonlinear systems.

\section{Computation of Resilience} 
\label{sec:4}
In this section, we demonstrate how to compute exactly both the resilience metric and the associated optimal controller, for the general class of specifications defined in Section~\ref{Sec:specfication}, with linear time-varying control systems, considering linear closed-loop controllers and open-loop controllers. For linear closed-loop controllers, the computation of the resilience metric is formulated as a polynomial optimization problem, whereas for open-loop controllers, it reduces to a linear optimization problem. Furthermore, we show how the resilience metric can be approximated up to any accuracy for nonlinear systems and nonlinear controllers.

\subsection{Time-varying linear systems and linear controllers}
% \subsubsection{Closed-loop controllers}
Consider the system $\Sigma$ in~\eqref{eqn:non linear controlled system} with an LTV dynamics, as in~\eqref{eqn:time varying linear sys} and a linear controller defined by \( \pi_\alpha(x(k)) =  u(k) = \alpha_1 x(k) + \alpha_2 \), where $(\alpha_1, \alpha_2) \in \mathbb{R}^{m\times n} \times \mathbb{R}^{m}$. 

Given a polytopic specification $\psi$ as defined in Equation~\eqref{polytopic_spec}, the calculation of the resilience metric requires the computation of the optimal parameters $\alpha_1$ and $\alpha_2$ that ensure the system respects the specification while simultaneously maximizing the disturbance set. The next corollary provides a closed-form expression of the resilience metric such that the state evolution satisfies $\psi$. 
\begin{corollary}
\label{theorem:1}
Consider the linear system $\Sigma$ in~\eqref{eqn:time varying linear sys} with a closed-loop controller $\pi_\alpha(x(k)) = \alpha_1 x(k) + \alpha_2$, and the specification $\psi$ as in~\eqref{polytopic_spec}. We have
\begin{equation}
\label{eq:resilience closed loop}
\begin{aligned}
g_\psi(x, \varepsilon_0) = &\max_{\mu \ge 0, \,  \alpha_1 \in \mathbb{R}^{m \times n}, \alpha_2 \in \mathbb{R}^{m}, P \geq0} \; \mu \\
s.t  \quad& P A_b = \mu \,E(\alpha_1) ,\\
&   P B_b \leq F(x, \alpha_1, \alpha_2, \varepsilon_0), \\
& P \geq 0,
\end{aligned}
\end{equation}
where matrices $A_b, B_b, E^{cl}(\cdot), F^{cl}(\cdot)$ are defined in \eqref{Ab_Bb_closed_1} and \eqref{E_F_closed_1} respectively. 
\end{corollary}
\begin{proof}
% From Definition~\ref{def:controlled_resilience}, the resilience metric for the closed-loop controller is
% \[
% \begin{aligned}
% g_\psi(x, \varepsilon_0) = &\max_{\mu \ge 0} \; \mu \\
% \text{s.t.} \;\; & \exists \alpha_1 \in \mathbb{R}^{m \times n}, \alpha_2 \in \mathbb{R}^{m},\\
% &\forall d(0), \dots, d(N-1) \in \Omega_{\mu}(0),\\
% &G_k x(k) \leq H_k, \hspace{0.5cm} \text{for } k = 0,\dots,N,\\
% & \alpha_1 x(j) + \alpha_2 \in \Omega_{\varepsilon_0}(0), \; \hspace{0.5cm} \text{for } j = 0,\dots,N-1.
% \end{aligned}
% \]
% To establish the equivalence, one follows the same steps of the proof of Theorem~\ref{Theorem:7} with $w_1 =  1$, $w_2 = 0$ and a fixed input bound $\varepsilon =  \varepsilon_0$.
By Definition~\ref{def:controlled_resilience}, the closed-loop resilience metric equals \eqref{eq:resilience closed loop} via Theorem~\ref{Theorem:7} with $w_1 = 1$, $w_2 = 0$, and $\varepsilon = \varepsilon_0$.
\end{proof}

\begin{remark}
Corollary~\ref{theorem:1} generalizes the results in \cite{ait2025maximal} in two significant directions. First, it takes into account input constraints and finds a new way to apply Farkas' lemma by grouping all state and input constraints in a single matrix inequality. This property maintains the equivalence when applying Farkas’ Lemma \cite{Alexander1999} and, importantly, does not increase the computational complexity. We preserve the problem structure, and the constraints of the resulting optimization problem remain linear in $\mu$ and polynomial in $\alpha_1$ and $\alpha_2$. 
Second, the proposed framework strictly generalizes existing results established for time-invariant linear systems in \cite{ait2025maximal}. These results arise as a special case of the present work and can be recovered directly by considering $A_k = A$ and $B_k = B$ for all $k \in \N$. In addition, we extended the results of \cite{saoud2024} from calculating the resilience for linear time-invariant autonomous systems to a wider class of time-varying linear control systems governed by linear state-feedback control under a more general class of specifications.
\end{remark}

% \begin{remark}
% This definition of resilience allows us to directly recover the formulation without input constraints from \cite{ait2025maximal} by either assigning a sufficiently large value to $\varepsilon_0$ or by simply removing the input constraints from the optimization problem. Notably, the resilience metric in the absence of input constraints serves as an upper bound for the resilience metric where these constraints are considered.
% \end{remark}

% \subsection{Linear systems with open-loop controllers}
% Consider the system $\Sigma$ in~\eqref{eqn:non linear controlled system} with a time varying linear dynamics, as in~\eqref{eqn:time varying linear sys}. 

The next corollary shows how to compute the resilience metric for the case of time-varying linear systems, polytopic specifications, and open-loop controller.
\begin{corollary}
\label{theorem:2}
Consider the linear system $\Sigma$ in~\eqref{eqn:time varying linear sys} with an open-loop controller defined by a sequence of inputs $\uu = (u_0,u_1,\dots, u_{N-1}) \in U^N$, and the specification~$\psi$ as in~\eqref{polytopic_spec}. We have
\begin{equation}
\label{eq:resilience open loop}
\begin{aligned}
g_\psi(x, \varepsilon_0) = &\max_{\mu \ge 0, \,  \mathbf{\tilde{u}} \in \Omega_{1}(0)^N, P \geq0} \; \mu \\
s.t  \quad & P A_b =\mu \; E^{ol} ,\\
& P B_b \leq F^{ol} (x, \mathbf{\tilde{u}} , \varepsilon_0), \\
& P \geq 0,
\end{aligned}
\end{equation}
where matrices $A_b, B_b, E^{ol} (\cdot), F^{ol} (\cdot)$ are defined in \eqref{Ab_Bb_closed_1} and \eqref{E_F_open_loop_1} respectively. 
\end{corollary}
\begin{proof}
% By Definition~\ref{def:controlled_resilience}, the resilience metric for the open-loop controller is
% \begin{equation}
% \begin{aligned}
% g_\psi(x, \varepsilon_0) = &\max_{\mu \ge 0} \; \mu \\
% \text{s.t.} & \quad \exists  u(0), \ldots, u(N-1) \in \Omega_{\varepsilon_0}(0), \\
% &\forall d(0), \dots, d(N-1) \in \Omega_{\mu}(0),\\
% &G_kx(k) \leq H_k, \text{ for } k  = 0, \dots,N.\\
% \end{aligned}
% \end{equation}
% To establish the equivalence we follow the same steps of the proof of Theorem~\ref{Theorem:8} with $w_1 =  1$, $w_2 = 0$ and a fixed input bound $\varepsilon =  \varepsilon_0$.
By Definition~\ref{def:controlled_resilience}, the open-loop resilience metric equals \eqref{eq:resilience open loop} via Theorem~\ref{Theorem:8} with $w_1 = 1$, $w_2 = 0$, and $\varepsilon = \varepsilon_0$.
% \begin{equation}
% \begin{aligned}
% g_\psi(x, \varepsilon_0) = &\max_{\mu \ge 0} \; \mu \\
% \text{s.t.} & \quad \exists  u(0), \ldots, u(N-1) \in \Omega_{\varepsilon_0}(0), \\
% &\forall d(0), \dots, d(N-1) \in \Omega_{\mu}(0),\\
% &G_kx(k) \leq H_k, \text{ for } k  = 0, \dots,N.\\
% \end{aligned}
% \end{equation}
\end{proof}

% The result of the previous corollaries enables the transformation of the resilience metric computation from a robust optimization problem into a deterministic optimization problem. This reformulation eliminates uncertainties and reduces the optimization problem to a linear program (Corollary~\ref{theorem:2}), which offers significant advantages over nonlinear formulations (Corollary~\ref{theorem:1}), including convexity, guaranteed global optimality and polynomial-time solution methods \cite{boyd2004convex,bertsekas2009convex, andersen2000mosek}.
% \begin{remark}
% To reduce the non-linearity of the optimization problem, we have used the following formulation

% \begin{equation}
% \begin{aligned}
% g_\psi(x, \varepsilon_0) = &\max_{\mu \ge 0} \; \mu \\
% \text{s.t.} & \quad \exists \mathbf{\tilde{u}} \in \Omega_{1}(0), P \geq 0, \\
% &P A_b = \boldsymbol{\mu} \; E ,\\
% & P B_b \leq F(x, \mathbf{\tilde{u}} , \varepsilon_0),\\
% \end{aligned}
% \end{equation}
% which is different from the formulation of Theorem 1 in \cite{ait2025maximal}. Within the proof of Theorem~\ref{Theorem:7} and Theorem~\ref{Theorem:8}, when establishing the equivalence via Farkas' Lemma, we incorporate $\mu$ as a component of the matrix $E$ rather than including it in $F$. This property is particularly valuable for applications involving open-loop controllers, as it reduces the problem into a linear program. This reformulation yields a significant computational advantage for the closed-loop controller. When $\alpha_1$ is fixed or zero, the optimization problem becomes linear.
% \end{remark}

The next section introduces the case of nonlinear systems and employs scenario optimization to solve the required robust optimization problem.

\subsection{Nonlinear systems and nonlinear controllers}
\label{resilience_scenario_section}

Consider the nonlinear control system introduced in~\eqref{eqn:non linear controlled system}. We define the scenario-based optimization problem for the resilience metric following the same steps and notation as in Section~\ref{scenario_optimization_steps} considering the case $w_1 = 1$, $w_2 = 0$, and $\varepsilon = \varepsilon_0$. We begin by formulating the problem with infinitely many constraints. Then, by sampling from the disturbance set, we obtain
\begin{equation}
\label{eqn:scenario_pb_rew_resilience}
\begin{aligned}
\sup_{(\mu, \alpha) \in \mathbb{R}_{\ge 0} \times \mathbb{R}^d} \; & \mu \\
\text{s.t. } \; & (\mu, \alpha) \in \bigcap_{i=1}^{M} \mathcal{X}_{\boldsymbol{\delta}_i},
\end{aligned}
\end{equation}
where the scenario constraint set \(\mathcal{X}_{\boldsymbol{\delta}_i}\) is defined using normalized disturbances \(\boldsymbol{\delta}_i \in \mathcal{D} = \Omega_1(0)^N\), \(i=1,\dots,M\), sampled i.i.d. according to any probability measure \(\mathbb{P}\) on \(\mathcal{D}\)
\begin{equation}
\begin{aligned}
\mathcal{X}_\delta = \big\{ (\mu, \alpha)& \in \mathbb{R}_{\ge 0} \times \R^d \mid \\
& \xi(x, \pi_\alpha, \mu\boldsymbol{\delta}) \in \psi \times \Omega_\varepsilon(0)^N   \big\}.
\end{aligned}
\end{equation}
The corresponding optimality guarantee is given in Theorem~\ref{theorem:9} using the optimization problem in~\eqref{eqn:scenario_pb_rew_resilience}.

% \begin{corollary}
% \label{Theorem:3}
% Consider the system $\Sigma$ in~\eqref{eqn:non linear controlled system}. For $x \in X$, a specification $\psi \subseteq X^{N+1}$ as defined in \ref{Sec:specfication} and an input set $U \subseteq \R^m$, let $\theta^*_M$ be the solution to the optimization problem defined in~\eqref{eqn:scenario_pb_rew_resilience} for $M \in \mathbb{N}$ number of scenarios. For any probability measure \(\mathbb{P}\), any confidence \(\beta \in (0,1]\), and with \(b(k)\) defined as in~\eqref{eqn:epsilon} for \(k=0,\dots,M\), it holds that
% \[
% \mathbb{P}^M \big( V(\theta_M^*) \le b(s_M^*) \big) \ge 1 - \beta.
% \]
% \end{corollary}
\section{Computation of Effort Metric}
\label{sec:5}

In this section, we demonstrate how to compute the effort metric as well as the corresponding optimal controller for a predefined level of disturbance for the system. As in the previous section, we begin with LTV systems and linear controllers where we provide the exact solution and then a probabilistic solution for the case of nonlinear systems and nonlinear controllers.
\subsection{Time-varying linear systems and linear controllers}
% \subsubsection{Closed-loop controllers}
% This section addresses the computation of controllers maximizing the resilience metric for linear discrete-time systems. 
Consider the system $\Sigma$ in~\eqref{eqn:time varying linear sys}. In the next corollary, building on the same system and specifications in the previous section, we demonstrate that the effort metric can be computed by formulating it as a standard optimization problem. 
% The first theorem addresses the closed loop controller, while the second focuses on the open loop controller.  

\begin{corollary}
\label{Theorem:4}
 Consider the linear system $\Sigma$ in~\eqref{eqn:time varying linear sys} with a linear controller $\pi_\alpha(x(k)) = \alpha_1 x(k) + \alpha_2$, and the specification $\psi$ as in~\eqref{polytopic_spec}. We have
\begin{equation}
\label{eq:effort closed loop}
\begin{aligned}
h_\psi(x, \mu_0) = &\min_{\varepsilon \ge 0, \,  \alpha_1 \in \mathbb{R}^{m \times n}, \alpha_2 \in \mathbb{R}^{m}, P \geq0} \; \varepsilon \\
s.t  \quad& P A_b =\mu_0 E^{cl}(\alpha_1) ,\\
&   P B_b \leq F^{cl}(x, \alpha_1, \alpha_2, \varepsilon), \\
& P \geq 0,
\end{aligned}
\end{equation}
where matrices $A_b, B_b, E^{cl}(\cdot), F^{cl}(\cdot)$ are defined in \eqref{Ab_Bb_closed_1} and \eqref{E_F_closed_1} respectively. 
\end{corollary}
% \begin{proof}
% % By Definition~\ref{def:effort_metric_unified}, the effort metric for the closed-loop controller is given by
% % \begin{equation}
% % \label{eq:effort closed loop 2}
% % \begin{aligned}
% % h_\psi(x, \mu_0) = &\min_{\varepsilon \ge 0} \; \varepsilon \\
% % \text{s.t.} \;\; & \exists \alpha_1 \in \mathbb{R}^{m \times n}, \alpha_2 \in \mathbb{R}^{m},\\
% % &\forall d(0), \dots, d(N-1) \in \Omega_{\mu_0}(0),\\
% % &G_k x(k) \leq H_k, \hspace{0.5cm} \text{for } k = 0,\dots,N,\\
% % & \alpha_1 x(j) + \alpha_2 \in \Omega_{\varepsilon}(0), \; \hspace{0.5cm} \text{for } j = 0,\dots,N-1.
% % \end{aligned}
% % \end{equation}
% % The equivalence between the optimization problem in \eqref{eq:effort closed loop} and the optimization problem in~\eqref{eq:effort closed loop 2} can be established following the same steps as the proof of Theorem~\ref{Theorem:7} with $w_1 =  0$, $w_2 = 1$ and a fixed disturbance bound $\mu =  \mu_0$.
% By Definition~\ref{def:effort_metric_unified}, the closed-loop resilience metric equals \eqref{eq:effort closed loop} via Theorem~\ref{Theorem:7} with $w_1 = 0$, $w_2 = 1$, and $\mu = \mu_0$.
% \end{proof}

This corollary establishes that computing the minimum control effort for linear systems under linear state-feedback can be formulated as a deterministic optimization problem.

This result enables a tractable, exact synthesis procedure for effort-optimal controllers, converting an inherently robust optimization into a finite constraints polynomial program. As a consequence, it provides a practical design methodology that not only quantifies the minimum control effort required to satisfy performance specifications, but also guarantees robustness against bounded disturbances.

We also provide a closed-form expression of the effort metric for the case of time-varying linear systems, polytopic specifications, and open-loop controller.

\begin{corollary}
    \label{Theorem:5}
Consider the system $\Sigma$ in~\eqref{eqn:time varying linear sys} with a linear controller defined by a sequence of inputs $\uu = (u_0,u_1,\dots, u_{N-1}) \in U^N$, and the specification $\psi$ as in~\eqref{polytopic_spec}. We have
\begin{equation}
\label{eq:effort open loop}
\begin{aligned}
h_\psi(x, \mu_0) = &\min_{\varepsilon \ge 0, \,  \mathbf{\tilde{u}} \in \Omega_{1}(0)^N, P \geq0} \; \varepsilon \\
s.t  \quad &P A_b = \mu_0 E^{ol} ,\\
&P B_b \leq F^{ol}(x, \varepsilon, \mathbf{\tilde{u}}), \\
& P \geq 0,
\end{aligned}
\end{equation}
where matrices $A_b, B_b, E^{ol} (\cdot), F^{ol} (\cdot)$ are defined in \eqref{Ab_Bb_closed_1} and \eqref{E_F_open_loop_1}, respectively. 
\end{corollary}

This definition of effort metric captures the following intuition: it quantifies the minimum control input magnitude needed to ensure the system satisfies the given specification, when the system trajectory is subjected to disturbances $d(k) \in \Omega_{\mu_0}(0)$ for $k \in \N$. In the special case where disturbances are zero, and the nominal behavior does not already meet the specification, the effort metric answers the question of finding the minimal control effort required to make the system satisfy the specification, and the associated controller parameters $\alpha$ enabling this behavior. Any effort below this threshold would be insufficient for the system to meet the specification, rendering satisfaction impossible. Conversely, if the given perturbation bound is increased beyond the resilience metric without input constraint, the corresponding optimization problem becomes infeasible.

\begin{remark}
In this section, we considered the computation of the effort metric for a bound on the disturbances given by $\mu_0$. For the linear closed-loop controller, the constraints of the optimization problem are linear in $\varepsilon$ and polynomial in $\alpha_1$ and $\alpha_2$. For the open-loop controller the constraints of the optimization problem are polynomial in $\varepsilon$ and $\uuu$.
\end{remark}
\begin{remark}
To fully characterize a time-varying system coupled with an initial point $x$ and a specification $\psi$, we first compute the inherent resilience metric without input constraints to obtain the maximum tolerable disturbance $\mu_{\text{max}} = g_{\psi}(x_0)$. We then calculate the corresponding minimal required effort $\varepsilon_{\text{max}} = h_\psi(x_0, \mu_{max})$ under that worst-case disturbance $\mu = \mu_{\text{max}}$. Finally, we evaluate the effort metric in the absence of perturbation $\mu = 0$ to determine the baseline minimum effort $\varepsilon_{\text{min}} = h_{\psi}(x_0)$. This yields two fundamental operating intervals: the feasible disturbance range $[0, \mu_{\text{max}}]$ and the required effort range $[\varepsilon_{\text{min}}, \varepsilon_{\text{max}}]$. With this mapping, the system is completely characterized. For any chosen disturbance level $\mu \in [0, \mu_{\text{max}}]$, one can determine the minimal required effort $ h_\psi(x_0, \mu)$; conversely, for any given effort $\varepsilon \in [\varepsilon_{\text{min}}, \varepsilon_{\text{max}}]$, one can find the maximal allowable disturbance $g_\psi(x_0, \varepsilon)$.
\end{remark}

In the next section, we extend the computation of the effort metric to the case of nonlinear systems.

\subsection{Nonlinear systems and nonlinear controllers}
To define the scenario-based optimization problem for the effort metric, we follow the same procedure as in Section~\ref{scenario_optimization_steps}, considering the case \( w_1 = 0 \), \( w_2 = 1 \), and \( \mu = \mu_0 \). We begin by formulating the problem with infinitely many constraints. Then, by sampling from the disturbance set, we obtain
\begin{equation}
\label{eqn:scenario_pb_rew_effort}
\begin{aligned}
\inf_{(\varepsilon, \alpha) \in \mathbb{R}_{\ge0} \times \mathbb{R}^d} \; & \varepsilon \\
\text{s.t. } \; & (\varepsilon, \alpha) \in \bigcap_{i=1}^{M} \mathcal{X}_{\boldsymbol{\delta}_i},
\end{aligned}
\end{equation}
where the scenario constraint set \(\mathcal{X}_{\boldsymbol{\delta}_i}\) is defined using normalized disturbances \(\boldsymbol{\delta}_i \in \mathcal{D} = \Omega_1(0)^N\), \(i=1,\dots,M\), sampled i.i.d. according to any probability measure \(\mathbb{P}\) on \(\mathcal{D}\), as
\begin{equation*}
\mathcal{X}_{\boldsymbol{\delta}} = \left\{ (\varepsilon, \alpha)\in \mathbb{R}_{\ge 0} \times \mathbb{R}^d \; \middle| \; \xi(x, \pi_\alpha, \mu_0\boldsymbol{\delta}) \subseteq \psi \times \Omega_\varepsilon(0) \right\}.
\end{equation*}
The corresponding optimality guarantee is given in Theorem~\ref{theorem:9} using the optimization problem in~\eqref{eqn:scenario_pb_rew_effort}.

% \begin{corollary}
% \label{Theorem:6}
% Consider the system $\Sigma$ in~\eqref{eqn:non linear controlled system}. For $x \in X$, a specification $\psi \subseteq X^{N+1}$ as defined in \ref{Sec:specfication} and an input set $U \subseteq \R^m$, let $\theta^*_M$ be the solution to the optimization problem defined in~\eqref{eqn:scenario_pb_rew_effort} for $M \in \mathbb{N}$ number of scenarios. For any probability measure \(\mathbb{P}\), any confidence \(\beta \in (0,1]\), and with \(b(k)\) defined as in~\eqref{eqn:epsilon} for \(k=0,\dots,M\), it holds that
% \[
% \mathbb{P}^M \big( V(\theta_M^*) \le b(s_M^*) \big) \ge 1 - \beta.
% \]
% \end{corollary}
% % \begin{proof}
% % The result follows directly, since the original of problem~\eqref{eqn:scenario_pb_rew2} is a robust optimization problem, satisfying the required regularity conditions (Property 1 in \cite{garatti2024}).
% \end{proof}

% Thus, for a chosen confidence \(\beta\), a sufficiently large number of scenarios \(M\) guarantees, with high probability, that the solution \(\theta_M^*\) satisfies all but a fraction \(b(s_M^*)\) of the unseen disturbance sequences.

\section{Case studies}
\label{sec:6}
Numerical experiments were performed on a standard laptop equipped with an Intel® Core™ i7 processor running at 2.3 GHz, 32 GB RAM, and a 1 TB SSD. All simulations were executed on a 64-bit Windows 11 operating system using Python 3, on a single core without GPU acceleration.

\subsection{Mobile Robot}
We consider the linear dynamics of a mobile robot as in~\eqref{eqn:time varying linear sys}, where $A = B =\mathbb{I}_2$. The state is a two-dimensional vector characterized by the position $x \in\left[-1, 1.7\right] \times\left[0, 2\right]$ and the input vector $u \in  \R^2$ without constraints.
To describe the desired behavior of the system, we first define the three regions: $R_1= [-0.3, 0.3] \times [0.6, 1.25]$, $R_2 = [0.8, 1.5] \times  [ 1.2, 1.75]$ and $R_3 = [-1, 1.7] \times [0, 2]$.% as shown in Figure~\ref{fig:trajectory_nominal_max_dist}.
The desired behavior we want to satisfy is given by the $LTL_f$ formula
\begin{equation}
\label{spec:psi}
    \psi = \nex^2 R_1 \wedge  \square^{[4,6]} R_2 \wedge \square^{[0,6]} R_3, 
\end{equation}
which can be described as follows: remain in region $R_3$ from the start until step $6$ ( $\square^{[0,6]} R_3$), while reaching region $R_1$ at time step $2$ ($\nex^2 R_1$) and staying in region $R_2$ between time instances $4$ and $6$ ($\square^{[4,6]} R_2$).
The input is controlled using the closed-loop linear controller $\pi_\alpha(x(k)) = u(k) = \alpha_1 x(k) + \alpha_2$, where $(\alpha_1, \alpha_2) \in \mathbb{R}^{2\times 2} \times \mathbb{R}^{2}$ are the controller parameters. All trajectories in this example start from the same initial state $x(0) = (0, 0.2)$. 

\textbf{Resilience:} after solving the optimization problem outlined in Corollary~\ref{theorem:1}, we obtain the following values for the controller parameters: 
$$
\alpha_1 = \begin{pmatrix} -0.99 & 1.62 \\ -0.11 & -0.33 \end{pmatrix}, \quad \alpha_2 = (-1.028, 0.574),
$$
and the corresponding maximum disturbance is $g_{\psi}(x(0)) = 0.0686$. These values are obtained by solving the optimization problem using Python library Pyomo and Knitro solver \cite{byrd2006knitro} implemented based on the {interior-point methods} \cite{NumericalOp}.

\textbf{Effort metric: }in Figure~\ref{fig:pareto_curve}, the red point represents the minimal effort $h_\psi(x(0)) = 0.25$ for the nominal dynamical system $\mu_0 = 0$ which corresponds to the minimal input bound required for the system to satisfy the specification. The corresponding controller values are
$$
\alpha_1 = \begin{pmatrix} -0.99 & 0.99 \\ -7.73 \times10^{-10} & 3.05\times10^{-10} \end{pmatrix}, \quad \alpha_2 = (-0.15,  0.25),
$$
obtained after solving the optimization problem in Corollary~\ref{Theorem:4}. The minimal effort required for the maximal resilience is $\varepsilon_{max} = h_\psi(x(0), g_\psi(x(0)) = 1.001$.

\textbf{Resilience-effort metric}: The blue points on the curve in Figure~\ref{fig:pareto_curve} represent the optimal solutions obtained by solving the resilient-effort optimization problem described in Theorem~\ref{Theorem:7} for different values of $w_1$ and $w_2$. For example, with $w_1 = 0.5$ and $w_2 = 0.05$, the Pareto optimization returns a resilience value $\mu = 0.053$ and an effort value $\varepsilon = 0.559$, which corresponds to the black dot on the same curve.

This curve represents the Pareto frontier, which shows the precise boundary between achievable and unachievable resilience–effort pairs and describes the fact that achieving higher resilience requires higher control effort.

% . To generate this curve, we fix values of disturbances satisfying $\mu_0 < g_\psi(x(0))$ and calculate the effort metric $h_\psi(x(0), \mu_0)$. Note that the same curve can be drawn using the input bound values $\varepsilon_0 \in \bigl[h_\psi(x(0)), \varepsilon_{\max}\bigr]$ and then calculate the corresponding resilience value $g_\psi(x(0),\varepsilon_0)$. 
The green region above the curve corresponds to feasible pairs $(\varepsilon, \mu)$ for which a controller can be synthesized to meet the desired resilience level using an effort greater than the minimal required. Any point in this region is therefore achievable, although not necessarily optimal. In contrast, the red region below the curve represents the infeasible domain, where no linear controller can simultaneously achieve the specified resilience and effort levels. 

\textbf{Open-loop controller}: Figure~\ref{fig:opel_loop_pareto_curve} presents the Pareto-curve for the system using an open-loop controller under the same specification and initial point as the previous example. The value of resilience computed by solving the optimization problem in Corollary~\ref{theorem:2} is $\mu_{max}= g_\psi(x_0) = 0.0458$. The minimal effort required to achieve it is $h(x_0, \mu_{max} )= 0.39$ and minimal effort required for making the nominal system satisfy the specification is $h(x_0) = 0.25$ both computed using Corollary~\ref{Theorem:5}. As expected, higher resilience values require higher control effort. Lower resilience values compared to the open-loop case indicate that the closed-loop controller steers the system more effectively toward resilient behavior.
 
% \begin{figure}[htbp]
%     \centering
%     \includegraphics[scale=0.5]{Journal version/Figures/beloved_2.png}
%     \caption{Illustration of two controlled trajectories of the robot starting from the initial state $x(0)=(0, 0.2)$ for $6$ time steps where both trajectories must satisfy the specification $\psi$ in Equation~\eqref{spec:psi}. The yellow trajectory represents the state evolution without disturbances $\varepsilon = 0$ while the blue trajectory represents the evolution with a disturbance exceeding resilience bounds $\varepsilon > g_\psi(x(0))$.}
%     \label{fig:trajectory_nominal_max_dist}
% \end{figure}

% \begin{figure}[htbp]
%     \centering
%     \includegraphics[scale=0.55]{Journal version/Figures/LTL_farkas_1002.png}
%     \caption{Illustration of controlled trajectories of the robot starting from the initial state $x(0)=(0, 0.2)$ for $6$ time steps under the specification $\psi$ in Equation~\eqref{spec:psi} with disturbances within the resilience bounds $\varepsilon \leq g_{\psi}(x(0))$.}
%     \label{fig:trajectory_within_disturbance}
% \end{figure}
% From a computation view point the algorithm tooks ??? s to converge while in closed loop it tooks ???s. 
\begin{figure}[htbp]
    \centering
    \includegraphics[scale=0.55]{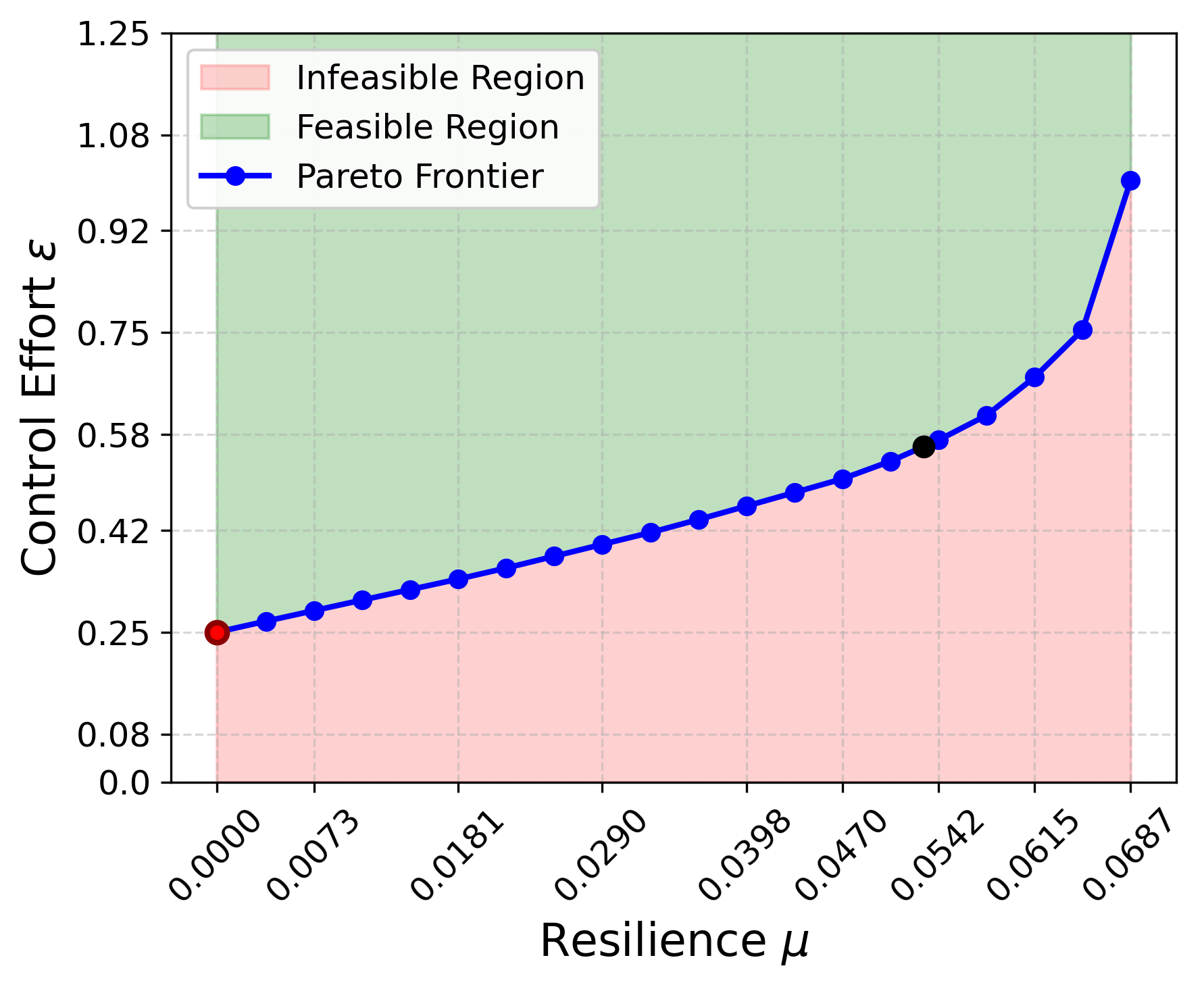}
    \caption{Illustration of the Pareto curve in blue which corresponds to the optimal trade-off for resilience and effort metric for controlled trajectories of the robot under the specification $\psi$ in Equation~\eqref{spec:psi}. The green region denotes the feasible set where a controller can be synthesized to achieve the effort-resilience level while the red region is the infeasible set where no controller can achieve the trade-off level.}
    \label{fig:pareto_curve}
\end{figure}

\begin{figure}[t]
    \centering
    \includegraphics[scale=0.55]{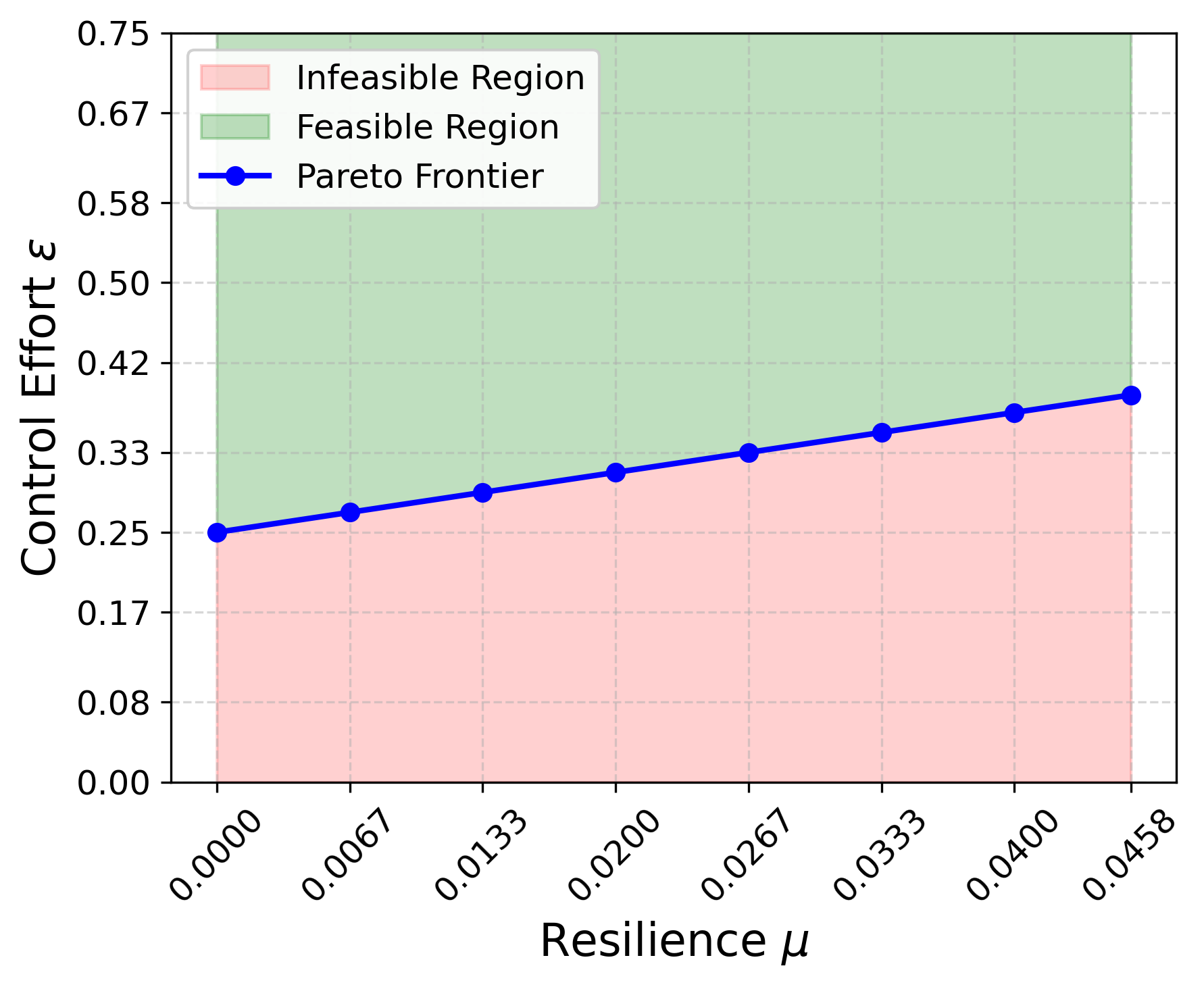}
    \caption{Illustration of the Pareto curve in blue which corresponds to the optimal trade-off for resilience and effort metric for controlled trajectories of the robot under the specification $\psi$ in Equation~\eqref{spec:psi}. The green region denotes the feasible set where an open-loop controller can be synthesized to achieve the effort-resilience level while the red region is the infeasible set where no open-loop controller can achieve the trade-off between effort and resilience.}
    \label{fig:opel_loop_pareto_curve}
    \vspace{-0.5cm}
\end{figure}

% \begin{figure}[H]
%     \centering
%     \includegraphics[scale=0.55]{Journal version/Figures/plot22.pdf}
%     \caption{ Pareto Frontier: Trade-off Between Robustness $\mu$ and Control Effort ($\varepsilon$)}
%     \label{fig:trajectory_within_disturbance}
% \end{figure}

% \begin{figure}[htbp]
%     \centering
%     \includegraphics[scale=0.8]{Figures/Farkas_linear.png}
%     \caption{Linear system with linear controller using theorem1}
%     \label{fig:trajectory_heatmap}
% \end{figure}

% Next, we added a linear controller to the dynamical system, with the input bounded between $-2$ and $2$. The problem was solved according to equation \eqref{eqn:scenario_pb}, using the previous values, the starting point $c_4 = (6, 6)$ and the matrix:

\subsection{Adaptive Cruise Control} 
\label{ACC}
Consider a vehicle following a lead car moving at a constant velocity $v_0$ on a straight road. The following vehicle adjusts its speed to maintain a safe distance while responding to changes in the environment or road conditions. The vehicle dynamics, adapted from \cite{ACC_Saoud}, are described by
\begin{equation}
\label{eq:ACC dynamics}
    \left\{\begin{aligned}
h(k+1) & =h(k)+\tau\left( \delta v_0 + v_0-v(k)\right), \\
v(k+1) & =v(k)+\frac{\tau}{m}\left(F(k)+ \delta f_0 -f_0-f_1 v(k)-f_2 v(k)^2\right),
\end{aligned}\right.
\end{equation}
where $v \geq 0$ represents the velocity of the vehicle, $h$ the distance between the lead and second vehicle, $m>0$ is the mass of the vehicle, and the term $f_0+f_1+f_2 v^2$ includes the rolling resistance and aerodynamics and $\tau$ represents a sampling period. The disturbance is modeled by $\delta v_0$, which is the uncertainty on the velocity $v_0$ of the lead vehicle and $\delta f_0$ is the uncertainty on the parameter
$f_0$. The variable $F$ represents the control input and must satisfy $F \in [\underline{F}, \overline{F}]$. The values of $\underline{F}$, $ \overline{F}$ and other constants are given in Table~\ref{fig:ACC_values}.

All trajectories in this example start from the same initial state $x(0) = (60, 15)$. We have conducted two experiments, the first using a linear controller defined for $x =(h, v) \in \R^2$ as $\pi_\alpha(x) = \alpha_1 x + \alpha_2$, where $\alpha_1 \in \R^2$ and $\alpha_2 \in \R$ are the parameters of the controller.
We consider the specification: $\psi=\psi_1 \wedge \psi_2$ with $\psi_1:=\bigcirc^3B_1$ and $\psi_2:=\nex^4 B_2$. This specification can be interpreted as follows: the relative position and velocity should reach the set $B_1= \mathcal{B}_{\sqrt{0.1}}((58.75, 16.4))$ in $3$ steps and reach the set $B_2= \mathcal{B}_{\sqrt{0.1}}((57.75, 15.6))$ in $4$ steps, which means that we will force the velocity to increase, which results in the decrease of the distance between the two cars. The objective is to compute the resilience metric under which the trajectory of the system satisfies the specification $\psi$.

The dynamical system in this use case is non-linear, and we used a scenario approach described in the Section~\ref{resilience_scenario_section}. We obtain the solution of the optimization problem for a set of $M$ i.i.d. scenarios $M$ sampled from uniform probability measure on the space of disturbances $\mathcal{D}$. This optimization is solved using the Python library Pyomo and the Knitro solver, which is based on interior-point methods \cite{NumericalOp}.

The optimal solution, the complexity $s_M^*$, and the violation bound $b$ are given in Table~\ref{tab:scnario_table} for different number of scenarios $M$ and different confidence values $\beta$.
One can see that increasing the number of scenarios decays the value of the resilience, which is expected since exploring more scenarios tends to include more constraints and allows for a tighter approximation of the resilience metric. To explain the values in Table~\ref{tab:scnario_table}, let us take the column of the number of scenarios $M = 500$. The corresponding value of the complexity is $s^*_M = 9$, which corresponds to the number of support constraints. Then, to calculate the bound $b$, we choose a level of confidence $\beta$. Taking $ \beta = 10^{-2}$, we can now calculate the bound $ b(s_M^*) = 0.046$. Hence, the probability of violation is bounded by $0.046$ with confidence $1-10^{-2}$. 
The probability guarantee provided by Theorem~\ref{theorem:9} is 
$
\mathbb{P}^M \big(V(\theta_M^*) \le b(s_M^*)\big) \ge 1 - 10^{-2} = 0.99
$, which means that we are requiring a $99\%$ confidence level that the violation probability is below the value \( b(s_M^*) \). In this case, there is only a $4.6\%$ chance that a new scenario constraint will not be satisfied by the solution \( \theta_M^* \).
The relationship between the bound \( b \) and the confidence level \( \beta \) is inversely proportional, as demonstrated in the table. For \( M = 10 \), the bound on the violation probability is very large \( 0.851 \). However, as we increase the number of scenarios to \( M = 500 \), the bound on the violation probability decreases, reaching a reasonable value equal to \( 0.046 \). This indicates that by increasing the number of scenarios, we are able to obtain a tighter and more reliable bound on the violation probability, effectively reducing the likelihood of a scenario violating the solution. The value of effort metric for $M= 100$ scenarios is $h_\psi(x(0))=  2295.55 \,kg\,m/s^2$.

    Figure~\ref{fig:linear_ACC} shows trajectories with disturbance less than the resilience metric $g_{\psi}(x(0)) = 0.036$ computed for $M = 100$. One can readily see that the trajectories reach the target set $B_1$ (in light blue) in three time steps and then reach the target set $B_2$ (in light red) in four time steps. At four time steps, we observe that only few points fall outside the set $B_2$, which confirms that, while our solution is not exact, it is probabilistically guaranteed. The values of the variable $\delta v_0$ and $\delta f_0$ for maximal resilience metric results correspond to the following intervals $v_0 = [-0.072, 0.072]$ 
and $f_0 = [-98.36,98.36]$.
%
%$f_0 = [−82.2, 82.2]$
%and $f_0 = [−82.2, 82.2]$.
%

\begin{figure}[t]
     \centering
     \includegraphics[width=0.5\textwidth]{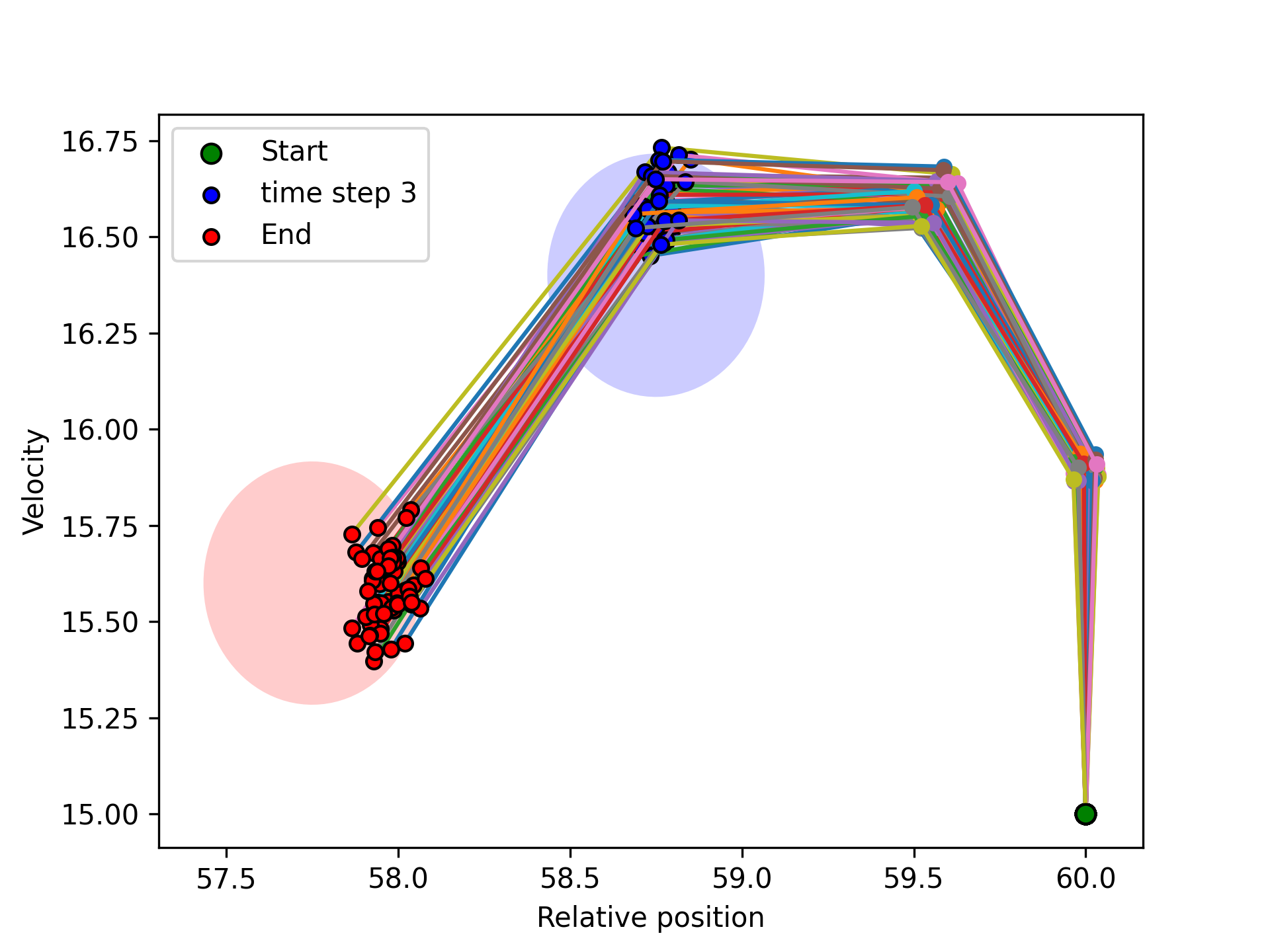}
   \caption{Sample trajectories for adaptive cruise control from Section~\ref{ACC} with the optimal linear controller and with disturbances less than the maximum disturbance given by the resilience metric $\varepsilon \leq g_{\psi}(x(0))=0.03$ for $M =100$ scenarios.}
\label{fig:linear_ACC}  
\vspace{-0.5cm}
\end{figure}
\begin{figure}[t]
     \centering
     \includegraphics[width=0.5\textwidth]{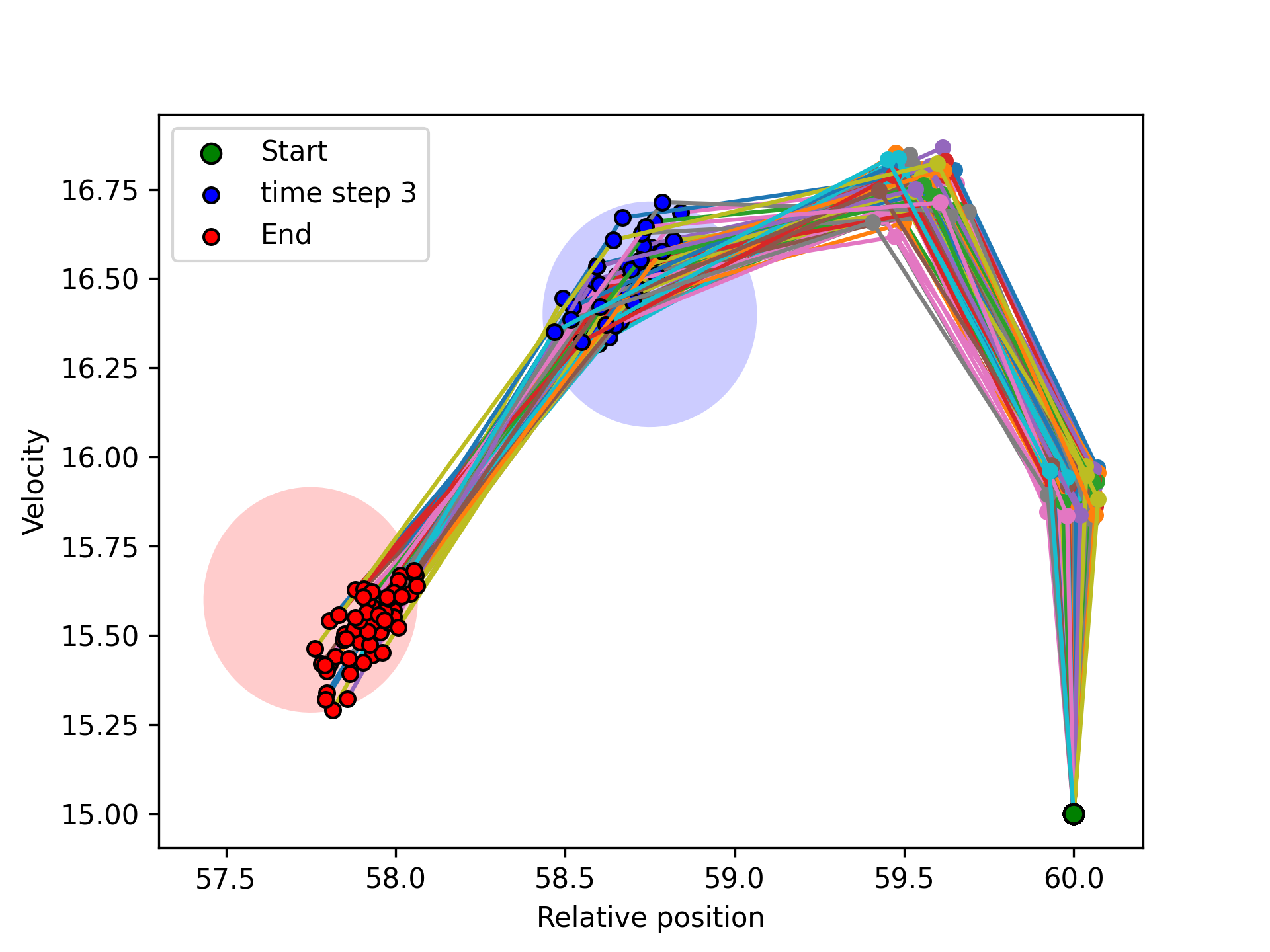}
     \caption{Sample trajectories for adaptive cruise control from Section~\ref{ACC} with the optimal polynomial controller and with disturbances less than the maximum disturbance given by the resilience metric $\varepsilon \leq g_{\psi}(x(0))=0.045 $ found for $M =100$ scenarios.}
\label{fig:polynomial_ACC}
\vspace{-0.5cm}
\end{figure}

% \vspace{-6.5em}
\begin{table}[t]
\centering
\renewcommand{\arraystretch}{1.5}
\setlength{\tabcolsep}{4pt} % Reduced column padding
\begin{tabular}{|c|c|c|c|c|c|c|}
\hline
Parameter & $m$ & $f_0$ & $f_1$ & $f_2$ & $\underline{F}$ & $\overline{F}$ \\ \hline
Value & 1370 & 51.0709 & 0.3494 & 0.4161 & -4031.9 & 2687.9 \\ \hline
Unit & kg & N & $\mathrm{N\,s/m}$ & $\mathrm{N\,s^2/m^2}$ & $\mathrm{kg\,m/s^2}$ & $\mathrm{kg\,m/s^2}$ \\ \hline
\end{tabular}
\caption{Parameters for the adaptive cruise control system~\cite{ACC_Saoud}.}
\label{fig:ACC_values}
\vspace{-0.3cm}
\end{table}
\begin{table}[t]
\centering
\resizebox{0.7\textwidth}{!}{% %
\begin{tabular}{|c|c|c|c|}
\hline
\diagbox[width=3cm, height=1.2cm]{variables}{ scenarios (M)}& $10 $ & $100$ & $500$\\
\hline
$\varepsilon^*_M$ & $0.039$& $0.0367$& $0.0305$\\
\hline
$\alpha_1^*$ &  $[23428.089 ,-567.85]$&   $[3377.689, -599.61]$& $[3426.05,-588.90] $\\
\hline
$\alpha_2^*$ &  $-194479.59$&  $-190979.28$&   $ -194041.69$\\
\hline
$s_M^*$ &  $4$&  $8$&   $9$\\
\hline
$b(s_M^*, \beta = 10^{-2})$&  $0.851$&  $0.202$& $0.046$\\
\hline
$b(s_M^*, \beta = 10^{-4})$& $0.936$& $0.259$& $0.059$\\
\hline
$b(s_M^*, \beta = 10^{-6})$& $0.971$& $0.307$& $0.072$\\
\hline
\end{tabular}}
\caption{Values for the optimal solution of the scenario optimization defined in~\eqref{eqn:scenario_pb_rew_resilience}, the complexity $s_M^*$ and the bound risk $b$ for different values of the number of scenarios $M$ and confidence $\beta$. }
\label{tab:scnario_table}
\vspace{-0.5cm}
\end{table}

To illustrate the case of nonlinear controllers, we have defined a polynomial controller of degree $2$ defined for $(h, v) \in \R^2$ as $\pi_\alpha(h, v) = \alpha_1 h^2 + \alpha_2 v^2 + \alpha_3 h v + \alpha_4 h + \alpha_5 v+ \alpha_6$. The objective is to design the parameters to maximize the resilience of the system under the same specification $\psi$.
The resulting trajectories in Figure~\ref{fig:polynomial_ACC} reach the target set $B_1$ (in light blue) in three time steps and reach the target set $B_2$ (in light red) in four time steps except for a small number of states as discussed in the previous example. The resulting value of the resilience metric for $M = 100$ is $g_{\psi}(x(0)) = 0.078$, which is a higher value for resilience under the linear controller. This indicates that the polynomial controller can force the system toward a more resilient behavior. The corresponding controller parameters are $\boldsymbol{\alpha}  = [\alpha_1, \alpha_2,\, \alpha_3, \,\alpha_4,\, \alpha_5,\, \alpha_6] = 
[-501.63,\; -997.46,\; 2142.48,\; 27197.67, \; -97858.42,\\ -59225.81]
.$
\subsection{Collision-Avoidance Scenario}

In a typical urban-intersection scenario, we consider an ego vehicle approaching from the southern direction and traveling northbound along a straight road. Simultaneously, an intruder vehicle approaches from the east, moving westbound at a constant velocity of 15 m/s. Their trajectories form a potential conflict point where the two paths intersect at a right angle~(c.f. Figure~\ref{collision_scenario}). The ego vehicle’s motion is governed by the following equations:
\begin{equation} 
\label{acc_modified}\left\{ \begin{aligned} r_x(k+1) &= r_x(k) - \tau \,v_L , \\[4pt] r_y(k+1) &= r_y(k) + \tau\, v(k), \\[4pt] v(k+1) &= v(k) + \frac{\tau}{m}\Big( F(k)  - f_1 v(k) - f_2 v(k)^2 \Big), \end{aligned} \right. \end{equation}
\begin{figure}[h]
     \centering
     \includegraphics[width=0.4\textwidth]{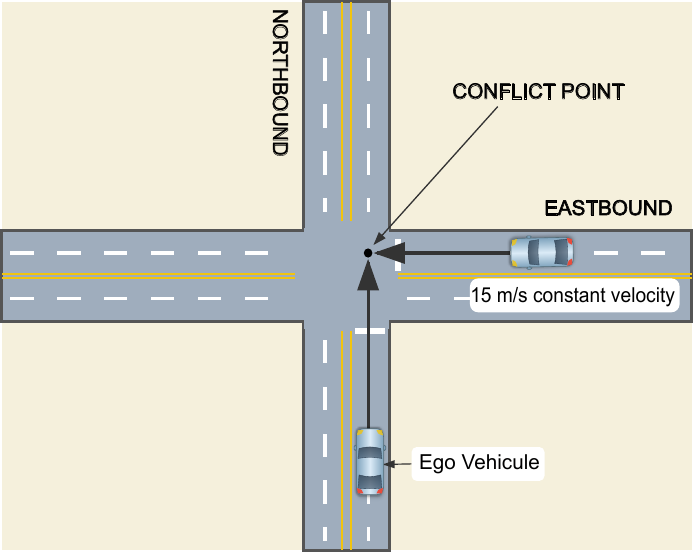}
     \caption{Illustration of the collision-avoidance scenario between the ego and the intruder cars.}
     \label{collision_scenario}
\end{figure}
Where $r_x$ and $r_y$ are the relative position of the two cars along the $x$-axis and the $y$-axis respectively. The safety specification means the Euclidean distance between the ego and intruder vehicles, $\mathsf{dist}(v, v_L)$, to remain at least than a minimum safe distance of $d_{\min} = 4m$ over a prediction horizon of seven time steps. Additionally, the ego vehicle must respect a prescribed velocity range: an upper bound to comply with traffic laws, and a lower bound to ensure the vehicle clears the intersection. Formally, this specification is given by $\psi =  \Box^7 D \wedge \Box^7 V$, where $ D = \{ (r_x, r_y, v) \in \mathbb{R}^3 \mid r_x^2 + r_y^2  \geq d_{min}\}$ and $ V = \{ (r_x, r_y, v) \in \mathbb{R}^3 \mid v_{\min} \le v \le v_{\max} \}$.
For the initial state $x_0 = (30, -27, 15)$ using a linear controller, the computed resilience for this specification is $g_{\psi}(x_0) = 1.12 \text{ m/s}$.
We evaluate four distinct scenarios to illustrate the system's behavior and the effectiveness of control strategies.
% \textcolor{red}{add nominal control fix at ...}
\begin{enumerate}
    \item Nominal Scenario (No Disturbance, No Control): The autonomous vehicle operates without external disturbances or an active collision-avoidance controller. The vehicles maintain a safe distance, and no clash occurs.

    \item Uncontrolled Scenario (Medium Disturbance): An additive disturbance $\mu =0.7m/s$ is introduced for the velocity. The disturbance causes the vehicle's trajectory to deviate, resulting in a collision and highlighting the necessity of a control strategy.

 \item  Active Control (Medium Disturbance): The same disturbance is applied, but the autonomous vehicle now employs a collision-avoidance (resilient) controller. The controller increases the acceleration to the maximum value, which results in increasing the vehicle's speed to create a larger separation. The speed remains elevated to maintain the safe distance. 

% 4.  Active Control (Maximum Tolerable Disturbance): A disturbance equal to the resilience bound $g_{\psi}(x_0) = 2.12 \text{ m/s}$ is applied while the controller remains active. The controller reacts aggressively by increasing the acceleration at its maximum value and maintaining this high control effort. Although this action would normally increase the vehicle’s velocity, the applied disturbance is strong enough, as a result, the overall velocity decreases. A collision eventually occurs, illustrating the limit of the formal safety guarantees. 
\item Minimal-Effort Control: The considered nominal behavior in this case does not satisfy the specification and the minimal-effort controller applies just enough acceleration to satisfy the safety specification and prevent collision. The vehicle’s speed increases slightly reflecting that only the minimal input is used.
\end{enumerate}

This analysis shows how the resilience metric $g_{\psi}$ characterizes the maximal allowable perturbation of the safety specification and allows the synthesis of controllers that ensure collision avoidance. By evaluating both effectiveness and efficiency, the effort metric enables the selection of control strategies that guarantee the specification while minimizing energy consumption.

Videos of the four scenarios using Carla simulator~\cite{CARLA} are available at this link: \href{https://youtu.be/soael9YJFFg}{https://youtu.be/soael9YJFFg}.
\subsection{Time-varying synchronous generator system}
In this example, we illustrate our theoretical findings on a discrete-time LTV system representing the sampled and linearized swing dynamics of a synchronous generator with an automatic voltage regulator (AVR). The system is written as in Equation~\eqref{eqn:time varying linear sys}, where the state vector is $x(k) = [\delta(k),\, \omega(k),\, e_{fd}(k)]^{\top}$, containing the rotor angle deviation $\delta(k)$ (rad), rotor speed deviation $\omega(k)$ (rad/s), and internal field voltage $e_{fd}(k)$. The control input is $u(k) = [u_1(k),\, u_2(k)]^{\top}$, where $u_1(k)$ denotes the mechanical torque (governor) command and $u_2(k)$ the AVR reference voltage input. The disturbance term $d(k)$ represents unmodeled load or torque variations acting on the generator shaft.
The time-varying state matrix $A_k \in \mathbb{R}^{3\times3}$ and input matrix $B_k \in \mathbb{R}^{3\times2}$ are given by
\begin{equation*}
A_k = 
\begin{bmatrix}
1 & T_s & 0 \\[3pt]
-\frac{T_sK_p(k)}{2H} &1 -\frac{T_s D}{2H} & -\frac{T_s K_f}{2H} \\[3pt]
0 & 0 & 1-\frac{T_s}{T_e}
\end{bmatrix}, \quad
B_k =
\begin{bmatrix}
0 & 0 \\[3pt]
\frac{T_s}{2H} & 0 \\[3pt]
0 & \frac{T_s}{T_e}
\end{bmatrix}.
\end{equation*}
The parameters are chosen to be physically representative of a medium-size synchronous machine where $H = 3.0$~s, $D = 0.5$, $K_f = 0.8$, $T_s = 0.5$~s, $T_e = 0.4$~s, and a slowly varying stiffness coefficient $K_p(k) = 2.0 - 0.05k$ modeling gradual operating-point changes. This LTV system reflects the generator’s primary frequency and voltage regulation dynamics and is representative of models used in modern power-system control studies~\cite{kundur1994power,lehn1999discrete}.
% \textcolor{red}{$F_0$ should be defined !!}
The desired behavior of the system is to prevent rotor excessive angle $\delta(k)$ deviations and the rotor speed deviation $\omega(k) $ to be near zero (synchronous speed), while ensuring that the internal field voltage $e_{fd}(k)$ adjusts appropriately in order to maintain voltage regulation. 
Figure~\ref{fig:time_varying_state_trajectories} shows the desired behavior given by the $LTL_f$ formula specification \begin{equation}
\label{spec:time_varying}
\psi = \square^{[2,30]} R_1 \wedge \square^{[0,30]} R_2, 
\end{equation} with $R_1 = [-0.5, 0.5] \times [-0.1, 0.1] \times [-2, 2]$, $ R_2 = [-2.0, 10.0] \times [-2, 10.0] \times [-5, 5],$ 
starting from point $x_0  = (0.5, 0, 0)$ and for $N= 30$ time steps using the open-loop resilient controller with input bound $\varepsilon = 0.38$.

% \begin{figure}[H]
%     \centering
%     \begin{subfigure}[H]{0.48\textwidth}
%         \centering
%         \includegraphics[width=\linewidth]{/Figures/state_trajectories_synchronous.png}
%         \caption{State trajectories showing rotor angle $\delta$ (rad), rotor speed $\omega$ (rad/s), and field voltage $e_{fd}$ (p.u.) over the prediction horizon $N= 30$ for the nominal and the controlled trajectories of the synchronous model. The shaded regions indicate constraint bounds that must be satisfied to respect the Specification~\eqref{spec:time_varying} by the controlled trajectories with the same color.}
%         \label{fig:state_trajectories}
%     \end{subfigure}
%     \hfill
%     \begin{subfigure}[H]{0.48\textwidth}
%         \centering
%         \includegraphics[width=\linewidth]{Journal version/Figures/control_inputs_optimal.png}
%         \caption{The normalized and synthesized open loop controller values corresponding to inputs $u_1$ (mechanical torque) and $u_2$ (AVR reference voltage) of the resilience metric. }
%         \label{fig:control_inputs}
%     \end{subfigure}
%     \caption{Nominal and controlled trajectories for the synchronous generator system along with the values of the controller deployed. The robust optimization ensures constraint satisfaction for input bound of  $0.38$ without adding disturbances.}
%     \label{fig:time_varying}
% \end{figure}

\begin{figure}[t]
     \centering
     \includegraphics[width=0.5\textwidth]{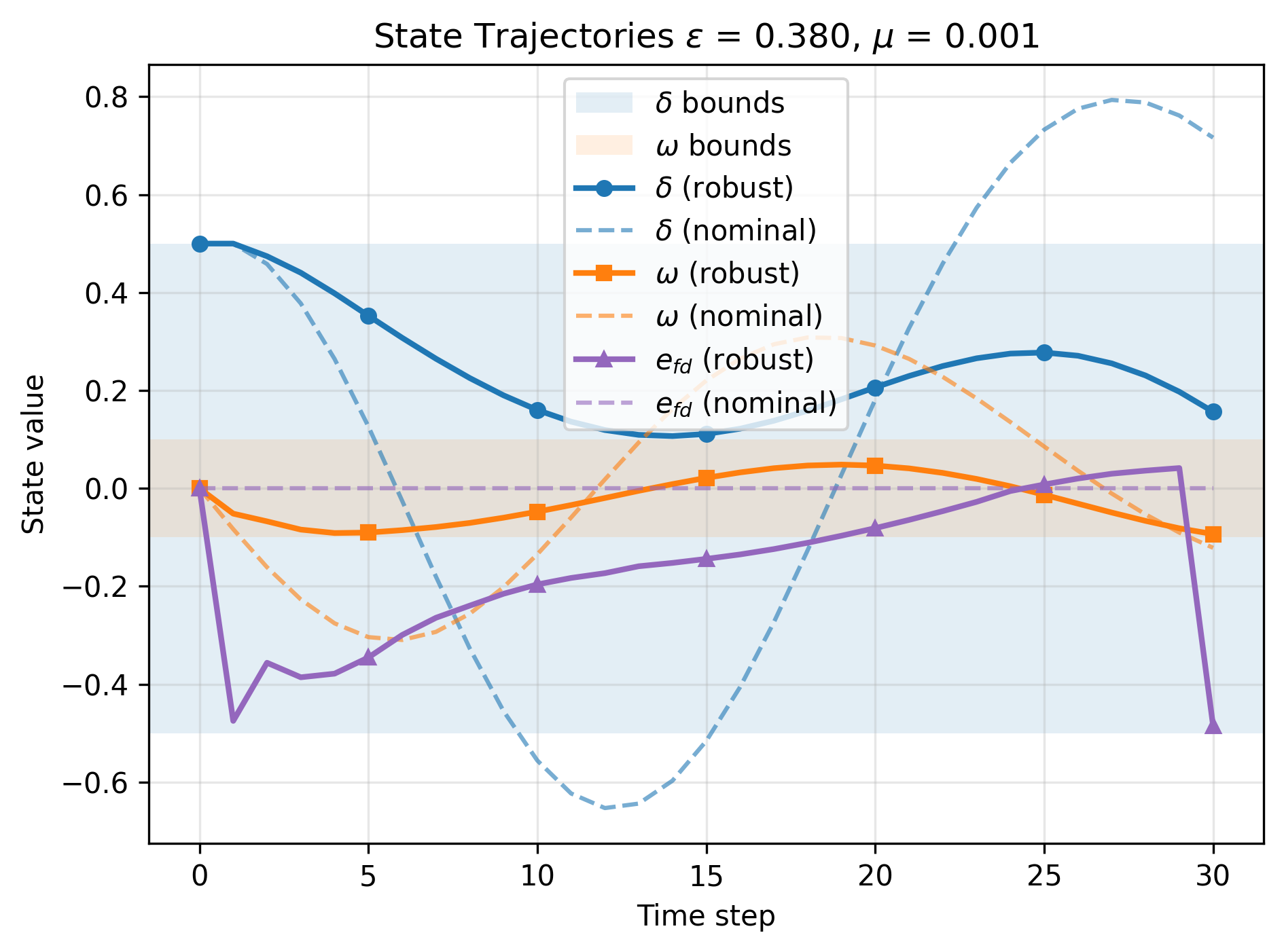}
     \caption{State trajectories showing rotor angle $\delta$ (rad), rotor speed $\omega$ (rad/s), and field voltage $e_{fd}$ over the horizon $N= 30$ for the nominal and the control trajectories of the considered model. The shaded regions indicate constraint bounds that must be satisfied to respect the specification $\psi$ in~\eqref{spec:time_varying} by the controlled trajectories with the same color.} \label{fig:time_varying_state_trajectories}
     \vspace{-0.3cm}
\end{figure}
\begin{figure}[t]
     \centering
     \includegraphics[width=0.5\textwidth]{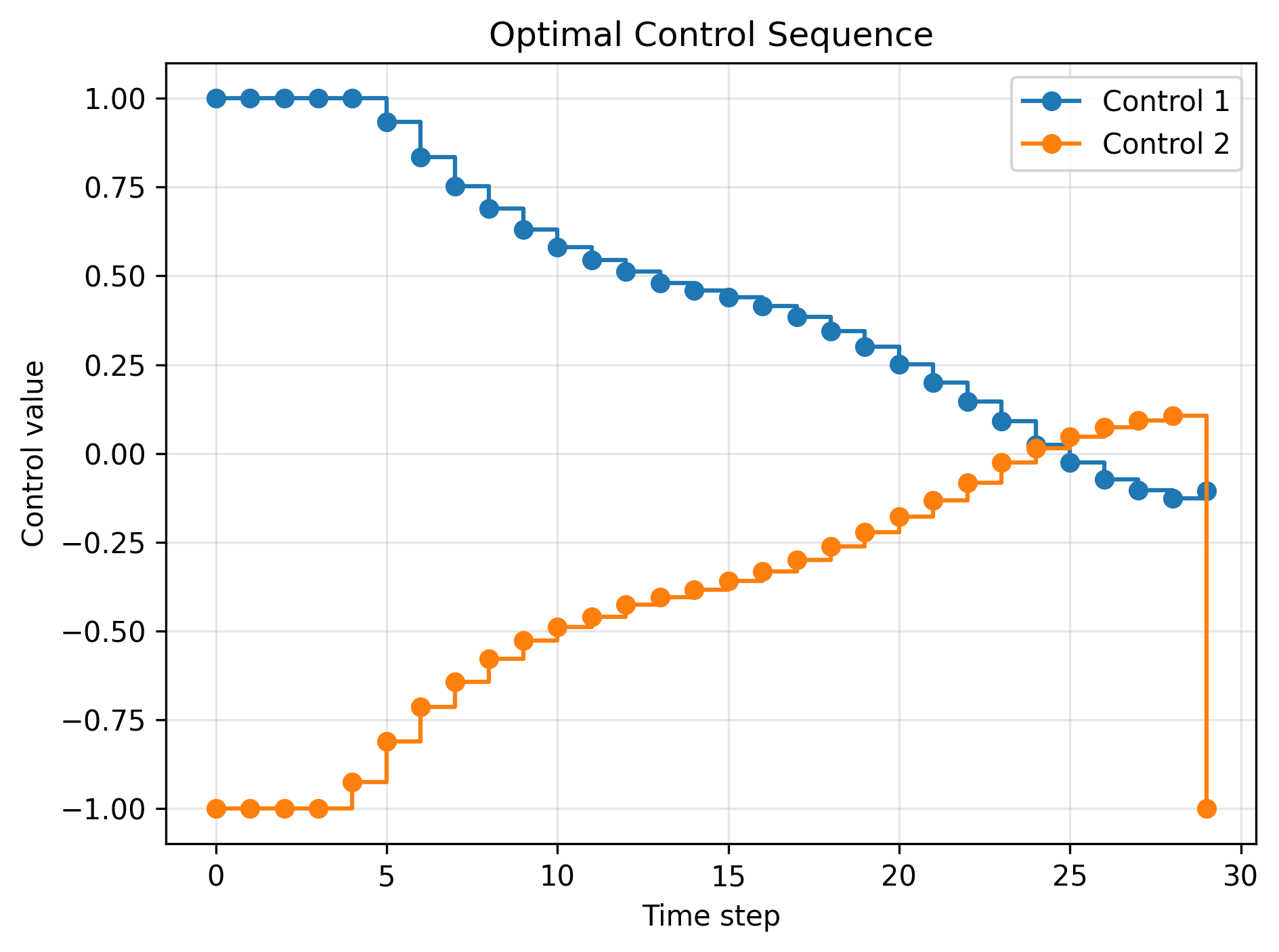}
     \caption{ The normalized and synthesized open-loop controller values corresponding to inputs $u_1$ (mechanical torque) and $u_2$ (AVR reference voltage) of the resilience metric. }
    \label{fig:control_inputs}
    \vspace{-0.5cm}
\end{figure}

The resilience metric computed using Corollary~\ref{theorem:2} is $ g_\psi(x_0) = 0.0031$. Its corresponding minimal effort $h_\psi(x_0, g_\psi(x_0)) = 0.397$ and the minimal effort metric $ h_\psi(x_0) = 0.367$ are both computed using Corollary~\ref{Theorem:5}. The computation time for resilience is $14s$ and the effort metric is $32 s$.
\section{Conclusion and future work}
\label{sec:conclusion}
% This paper is an initialization for resilience in controlled systems and can be a foundation for other studies like finding new way to solve the optimization problem using linear solver. Studying the properties of perturbed system with powerful tools to quantify the maximal tolerable disturbances, studying the controller synthesis of dynamical system while taking into account disturbances.
We provided a resilience metric framework for designing controllers for dynamical systems. For a given finite specification that defines the desired behavior of the system’s trajectory, this metric quantifies the optimal controller ensuring the system with input constraints satisfies the specification and the maximum disturbance for which the specification remain respected. In addition, we introduced a new \textit{effort metric} that quantifies the minimal input bound required to satisfy a given specification. For linear time variant dynamical systems and linear controllers, we demonstrated using Farkas' lemma how to compute the exact resilience and effort metric. In the general case of nonlinear systems and nonlinear controllers, a scenario approach applied to nonconvex optimization problems was used to derive these metrics while having a probabilistic guarantee on the solution. In future work, we aim to explore how the resilience and effort metric can be extended to interconnected control systems and continuous-time control dynamical systems.

% Limitation
% We cannot surpass 5 time step due to the explosion of running time for solving the optimization problem.
% Some trajectory controller diverge after a some time steps, and so we cannot apply any spcifcation atfer some horizon.
% the calculaton of the complexity increase exponontially with the increase of the number of scenarios. looking for more efficient algorithm to calculate it would be a promizing direction of research.

\bibliographystyle{IEEEtran}  % For IEEE style
% \section*{References}
\bibliography{references}     % Assumes your .bib file is "references.bib"
\appendix
\section{Auxilliary results}
The following theorem is a simple adaptation of the result in \cite[Corollary 7.lh]{Alexander1999}, and is known as the affine form of Farkas' lemma.
\begin{theorem}
\label{theorem farkas}
Suppose the set $\{x\mid E x \leq F\}$ is not empty. The following two statements are equivalent:
\begin{itemize}
    \item $Ex \leq F$ holds for all $x$ with $A x \leq b$;
    \item There exists a non-negative matrix $P$ such that $P A=E$ and $P b \leq F$.
\end{itemize}
\end{theorem}

\subsection{Proof of Theorem~\ref{Theorem:7}}

\begin{proof}
We have from the Definition~\ref{def:r_e_metric_unified} of resilience-effort metric that
\begin{equation}
\label{trade-off_def}
\begin{aligned}
p_\psi(x, w_1, w_2) = &\max_{\mu \ge 0, \varepsilon \ge 0} \; w_1 \mu - w_2 \varepsilon \\
\text{s.t.} \;\; & \exists \alpha_1 \in \mathbb{R}^{m \times n}, \alpha_2 \in \mathbb{R}^{m},\\
&\forall d(0), \dots, d(N-1) \in \Omega_{\mu}(0),\\
&G_k x(k) \leq H_k, \hspace{0.5cm} \text{for } k = 0,\dots,N,\\
& \alpha_1 x(j) + \alpha_2 \in \Omega_{\varepsilon}(0), \; \hspace{0.5cm} \text{for } j = 0,\dots,N-1.
\end{aligned}
\end{equation}

Letting $x = x(0)$, the solution of the closed-loop system at time $k \geq 1 $ is given by
\begin{equation}
x(k) =\underline{A}_{k,-1} x + \Big(\sum_{i = 0}^{k - 1} \underline{A}_{k,i} B_i\Big) \alpha_2 + \sum_{i = 0}^{k - 1}\underline{A}_{k,i}d_i,
\end{equation}
with $\underline{A}_{k,i} = \prod_{j=i+1}^{k-1}\bar{A}_j$ and the convention that an empty product equals the identity, i.e., $ \underline{A}_{k,k-1}=A_{1,0}=I_n$.

Let \(\alpha_1 \in \mathbb{R}^{m \times n}\) and \(\alpha_2 \in \mathbb{R}^{m}\). The state must satisfy \(G_k x(k) \leq H_k\) for \(k = 0,\dots,N\) and for all \(\mathbf{d} \in \Omega_{\mu}(0)^N\). For $k \in \{1, \dots, N \}$, substituting the expression of state trajectory \(x(k)\) yields
\[
G_k \left( \underline{A}_{k,-1} x + \sum_{i = 0}^{k - 1} \underline{A}_{k,i} B_i \alpha_2 + \sum_{i = 0}^{k - 1} \underline{A}_{k,i} d_i \right) \leq H_k.
\]
Rearranging terms to isolate the disturbance one gets

\[
G_k \sum_{i = 0}^{k - 1} \underline{A}_{k,i} d_i \leq H_k - G_k \underline{A}_{k,-1} x - G_k \left( \sum_{i = 0}^{k - 1} \underline{A}_{k,i} B_i \right) \alpha_2.
\]
This can be expressed compactly using matrix form as:
\[
E_{xk}(\alpha_1) Y \leq \frac{1}{\mu} F_{xk}(x, \alpha_1, \alpha_2, \varepsilon),
\]
where \(Y := \frac{1}{\mu} [0_{n\times 1},d_0, d_1, \ldots, d_{N-1}]^T \in \mathbb{R}^{n(N+1)}\), $E_{xk}$ and $F_{xk}$ are defined in~\eqref{Ex_Fx_closed_1}. The condition must hold for all $d_i$ such that $||d_i||_\infty \leq \mu$ which can be written in matrix form as $A_b Y \leq B_b$ for \(A_b\) and \(B_b\) defined in~\eqref{Ab_Bb_closed_1}.

Stacking these inequalities for \(k = 0, \dots, N\), we obtain

\begin{equation}
\label{state constraint closed 1}
E_x(\alpha_1) Y \leq \frac{1}{\mu} F_x(x, \alpha_1, \alpha_2),  \; \; \forall\, Y \; \;\text{such that} \; \; A_b Y \leq B_b,
\end{equation}
where \(E_x(\alpha_1)\), and \(F_x(x, \alpha_1, \alpha_2)\) are defined in~\eqref{Ex_Fx_closed_1}.
For the input constraints, the expression of the input at time $k \geq 0$ is
\begin{align*}
u(k) &= \alpha_1 x(k) + \alpha_2 \\
&= \alpha_1  \underline{A}_{k,-1} x 
+ \alpha_1 \sum_{i = 0}^{k - 1} \underline{A}_{k,i} B_i \alpha_2 + \\
& \qquad \alpha_1 \sum_{i = 0}^{k - 1} \underline{A}_{k,i} d(i) 
+ \alpha_2 \\
&= S_k(x, \alpha_1, \alpha_2) + \mu \alpha_1 E_k Y,
\end{align*}
where $S_k(x, \alpha_1, \alpha_2)$ and $E_k$ are defined in~\eqref{s_k} and \eqref{Ex_Fx_closed_1}.
% \[
% E_k = \left[0_n, \underbrace{\underline{A}_{k,0}, \underline{A}_{k,1}, \ldots, \underline{A}_{k, k-1}}_{k \text{ times}}, \underbrace{0_n,\dots, 0_n}_{N-k \text{ times}}\right] \in \mathbb{R}^{m \times n(N+1)}.
% \]
Therefore, the condition $\|u(k)\|_{\infty} \leq \varepsilon$ can be written in matrix form as
\begin{align*}
A_u u(k) \leq &\varepsilon B_u,\\
\iff A_u \mu E_k(\alpha_1) Y \leq&  \varepsilon B_u- A_u S_k(x, \alpha_1, \alpha_2)
\end{align*}
for $A_u$ and $B_u$ defined in~\eqref{A_u, b_u closed_1}. Stacking these inequalities for \(k = 0, \dots, N-1\), we obtain
\begin{equation}
\label{input constraint closed 1}
E_u(\alpha_1) Y \leq \frac{1}{\mu} F_u(x, \alpha_1, \alpha_2, \varepsilon), \; \; \forall\, Y \; \;\text{such that} \; \; A_b Y \leq B_b,
\end{equation}
where $E_u(\alpha_1)$ and $F_u(x, \alpha_1, \alpha_2, \varepsilon)$ are defined in~\eqref{E_u, F_u closed_1}.
Combining the matrix form of~\eqref{state constraint closed 1} and~\eqref{input constraint closed 1}, we get
\[
E^{cl} (\alpha_1)Y \leq \frac{1}{\mu} F^{cl}(x, \alpha_1, \alpha_2, \varepsilon), \; \; \forall\, Y \; \;\text{such that} \; \; A_b Y \leq B_b,
\]
where $E^{cl}(\alpha_1)  = [E_x(\alpha_1) , E_u(\alpha_1) ]^T$ and $F^{cl}(x, \alpha_1, \alpha_2, \varepsilon)  = [F_x(x,\alpha_1, \alpha_2), F_u(x, \alpha_1, \alpha_2, \varepsilon)]^T$.
Therefore, the feasibility conditions in \eqref{trade-off_def} can be written for $\alpha_1 \in \mathbb{R}^{m \times n}$ and $ \alpha_2 \in \mathbb{R}^{m}$ as
\begin{equation}
\label{matrix_form}
\begin{aligned}
 \, E^{cl}(\alpha_1) Y& \leq \frac{1}{\mu} F^{cl}(x, \alpha_1, \alpha_2, \varepsilon), \\
& \text{for all } Y \text{ satisfying } A_b Y \leq B_b \\
\iff\, \mu E^{cl}(\alpha_1) Y& \leq  F^{cl}(x, \alpha_1, \alpha_2, \varepsilon), \\
& \text{for all } Y \text{ satisfying } A_b Y \leq B_b   .
\end{aligned}
\end{equation}
Hence, using the affine form of Farkas' Lemma \cite{Alexander1999} (see Lemma~\ref{theorem farkas}) condition \eqref{matrix_form} holds if and only if there exists a nonnegative matrix $P \geq 0$ such that $P A_b= \mu E^{cl}(\alpha_1)$ and $P B_b \leq  F^{cl}(x, \alpha_1, \alpha_2, \varepsilon)$. Thus, the resilience-effort metric can be expressed as
\[
\begin{aligned}
p_\psi(x, w_1, w_2) = &\max_{\varepsilon \ge 0, \mu \ge 0} \; w_1\mu - w_2 \varepsilon \\
\text{s.t.} \;\; & \exists \alpha_1 \in \mathbb{R}^{m \times n}, \alpha_2 \in \mathbb{R}^{m}, P \geq 0,\\
& P A_b = \mu \,E^{cl}(\alpha_1) ,\\
&   P B_b \leq F^{cl}(x, \alpha_1, \alpha_2, \varepsilon_0), \\
& P \geq 0.
\end{aligned}
\]
Since the maximization is over all feasible $(\mu, \varepsilon, \alpha_1, \alpha_2, P)$, this is equivalent to the optimization problem in \eqref{eq:pareto_optimization}.
\end{proof}

\subsection{Proof of Theorem~\ref{Theorem:8}}

\begin{proof}
We have from the Definition~\ref{def:r_e_metric_unified} of resilience-effort metric that
\begin{equation}
\label{open_loop_trade_off}
\begin{aligned}
p_\psi(x, w_1, w_2) = &\max_{\mu \ge 0, \varepsilon \ge 0} \; w_1 \mu - w_2 \varepsilon \\
\text{s.t.} \;\; & \exists\uu \in \Omega_\varepsilon(0)^N,\\
&\forall d(0), \dots, d(N-1) \in \Omega_{\mu}(0),\\
&G_k x(k) \leq H_k, \hspace{0.5cm} \text{for } k = 0,\dots,N.\\
\end{aligned}
\end{equation}
The system is evolving according to
\begin{align*}
 x(k+1) &= A_kx(k) + B_ku(k) + d(k). 
\end{align*}
By denoting $x = x(0)$, the solution of the closed-loop system at time $k \geq 1$ is given by
\begin{equation}
\label{state_evolution_time_varying}
x(k) =  \prod_{i=0}^{k-1}A_i x + \sum_{i = 0}^{k - 1}\tilde{A}_{k,i}B_iu(i) + \sum_{i = 0}^{k - 1}\tilde{A}_{k,i}d(i),
\end{equation}
with $\tilde{A}_{k,i} = \prod_{j=i+1}^{k-1}A_j$ and the convention that an empty product equals the identity.
For a fixed normalized input sequence $\mathbf{\tilde{u}} = \frac{1}{\varepsilon}\uu \in \Omega_1(0)^N$, the condition $G_kx(k) \leq H_k, \text{ for } k  = 0, \dots,N$ must hold for all $\mathbf{d} \in \Omega_{\mu}(0)^N$, which can be expressed in matrix form as 
\[
E^{ol} Y \leq \frac{1}{\mu}F^{ol}(x, \varepsilon, \mathbf{\tilde{u}}) \;\; \forall\, Y \; \;\text{such that} \; \; \quad A_b Y \leq B_b,
\]
where $Y := \frac{1}{\mu}\left[0_{n \times 1},  d_0, d_1, \ldots, d_{N-1}\right]^T$ are free variables, and $A_b$, $B_b$, $F^{ol}(x, \mathbf{\tilde{u}}, \varepsilon)$ and $E^{ol}$ are defined in~\eqref{Ab_Bb_closed_1} and~\eqref{E_F_open_loop_1}, respectively. The inequality $A_b Y \leq B_b$ represents $\|Y\|_{\infty} \leq 1$ in matrix form. Therefore, the feasibility conditions in \eqref{open_loop_trade_off} can be written for $\mathbf{\tilde{u}} =  \Omega_1(0)^N$ as
\begin{equation}
\begin{aligned}
\, E^{ol} Y& \leq \frac{1}{\mu} F^{ol}(x, \uuu, \varepsilon), \\
& \text{for all } Y \text{ satisfying } A_b Y \leq B_b  \\
\iff  \, \mu E^{ol} Y &\leq  F^{ol}(x, \uuu, \varepsilon), \\
& \text{for all } Y \text{ satisfying } A_b Y \leq B_b   .
\end{aligned}
\end{equation}
Hence, using the affine form of Farkas' Lemma \cite{Alexander1999} (see Lemma~\ref{theorem farkas}) condition \eqref{matrix_form} holds if and only if there exists a nonnegative matrix $P \geq 0$ such that $P A_b= \mu E^{ol}$ and $P B_b \leq  F^{ol}(x, \uuu, \varepsilon)$. Thus, the set of achievable $(\mu, \varepsilon, \mathbf{\tilde{u}})$ triples is
\[
\mathcal{P} = \left\{ (\mu, \varepsilon, \mathbf{\tilde{u}}) \,\middle|\, 
\begin{aligned}
&\exists\, P \ge 0 \text{ such that } \\
&P A_b = \mu E^{ol},\;  P B_b \le F^{ol}(x, \varepsilon, \mathbf{\tilde{u}})
\end{aligned}
\right\}.
\]
This yields the optimization problem
\[
\max_{(\mu, \varepsilon, \mathbf{\tilde{u}}) \in \mathcal{P}} \; w_1\mu - w_2\varepsilon.
\]
Substituting the explicit characterization of $\mathcal{P}$ and jointly optimizing over all decision variables $(\mu, \varepsilon, \mathbf{\tilde{u}}, P)$ gives exactly \eqref{eq:pareto_open_loop}.
\end{proof}

\end{document}